\documentclass[oneside,reqno]{amsart}
\usepackage{amssymb}
\usepackage{eucal}
\usepackage{lscape}
\usepackage{url}
\usepackage{hyperref}
\usepackage{cite}
\hypersetup{colorlinks=true,urlcolor=blue}
\usepackage{graphicx}
\newtheorem{Theorem}{Theorem}

\newtheorem{Lemma}[Theorem]{Lemma}
\theoremstyle{remark}

\newcommand{\be}{\begin{equation}}
\newcommand{\ee}{\end{equation}}
\newcommand{\ea}[1]{\begin{align}#1\end{align}}
\newcommand{\nn}{\nonumber \\}
\newcommand{\bc}{\begin{center}}
\newcommand{\ec}{\end{center}}
\newcommand{\bmt}{\begin{pmatrix}}
\newcommand{\emt}{\end{pmatrix}}

\newcommand{\vpu}[1]{^{\vphantom{#1}}}

\newcommand{\la}{\langle}
\newcommand{\ra}{\rangle}
\newcommand{\fd}[1]{\mathbb{#1}}

\newcommand{\mbf}[1]{\mathbf{#1}}
\newcommand{\mcl}[1]{\mathcal{#1}}

\newcommand{\sy}{\hat{\rho}}
\newcommand{\ap}{\hat{\alpha}}
\newcommand{\syc}{\lambda}
\newcommand{\spj}{\mcl{P}_{AB}}
\newcommand{\syj}{\hat{\rho}_{AB}}
\newcommand{\eop}[1]{\hat{\epsilon}_{\rm{#1}}}
\newcommand{\dop}[1]{\hat{\delta}_{\rm{#1}}}
\newcommand{\eva}[1]{\Delta^{\sy}_{\rm{e}} #1}
\newcommand{\dva}[1]{\Delta^{\sy}_{\rm{d}} #1}
\newcommand{\uca}[1]{\Delta^{\sy} #1}
\newcommand{\evb}[1]{\Delta_{\rm{e}} #1}
\newcommand{\dvb}[1]{\Delta_{\rm{d}} #1}
\newcommand{\evc}[1]{\Delta^{\mcl{R}}_{\rm{e}} #1}
\newcommand{\dvc}[1]{\Delta^{\mcl{R}}_{\rm{d}} #1}
\newcommand{\ebva}[1]{\Delta^{\sy}_{\rm{D} \rm{e}} #1}
\newcommand{\dbva}[1]{\Delta^{\sy}_{\rm{D}\rm{d}} #1}
\newcommand{\ebvb}[1]{\Delta_{\rm{D}\rm{e}} #1}
\newcommand{\dbvb}[1]{\Delta_{\rm{D}\rm{d}} #1}
\newcommand{\ebvc}[1]{\Delta^{\mcl{R}}_{\rm{D}\rm{e}} #1}
\newcommand{\dbvc}[1]{\Delta^{\mcl{R}}_{\rm{D}\rm{d}} #1}
\newcommand{\etva}[1]{\Delta^{\sy}_{\rm{C}\rm{e}} #1}
\newcommand{\dtva}[1]{\Delta^{\sy}_{\rm{C}\rm{d}} #1}
\newcommand{\etvb}[1]{\Delta_{\rm{C}\rm{e}} #1}
\newcommand{\dtvb}[1]{\Delta_{\rm{C}\rm{d}} #1}
\newcommand{\etvc}[1]{\Delta^{\mcl{R}}_{\rm{C}\rm{e}} #1}
\newcommand{\dtvc}[1]{\Delta^{\mcl{R}}_{\rm{C}\rm{d}} #1}
\newcommand{\DM}{D }
\newcommand{\OP}{O }
\newcommand{\CO}{C }
\newcommand{\lX}{l_{\rm{X}}}
\newcommand{\lP}{l_{\rm{P}}}

\DeclareMathOperator{\Tr}{Tr}

\DeclareMathOperator{\SU}{SU}

\begin{document}
\title[Quantum Errors and Disturbances]{Quantum Errors and Disturbances:  Response to Busch, Lahti and Werner.}
\maketitle

%\vspace{0.1 cm}
\begin{center} D.M. Appleby
\\
\emph{Centre for Engineered Quantum Systems, School of Physics, The University of Sydney, Sydney, NSW, Australia}
 \end{center}

\vspace{0.5 cm}
 
\begin{center} \textbf{Abstract}

\vspace{0.5 cm}

\vspace{0.35 cm}
\parbox{12 cm }{Busch, Lahti and Werner (BLW) have recently criticized the operator approach to the description of quantum errors and disturbances.   Their criticisms are justified to the extent that the physical meaning of the operator definitions has not hitherto been adequately explained.   We rectify that omission.  We then examine BLW's criticisms in the light of our   analysis.  We argue that, although the approach BLW favour (based on the Wasserstein 2-deviation) has its uses, there are important physical situations where an operator approach is preferable.  We also discuss the reason why the error-disturbance relation is still giving rise to controversies almost a century after Heisenberg first stated his microscope argument.  We argue that the source of the difficulties is the problem of interpretation, which is not so wholly disconnected from experimental practicalities as is sometimes supposed.}
\end{center}

\vspace{0.2 in}

%\tableofcontents
\allowdisplaybreaks
\section{Introduction}
The error-disturbance principle remains highly controversial almost a century after Heisenberg wrote the paper~\cite{Heisenberg:1927} which originally suggested it.  It is  remarkable that this should be so, since the disagreements concern what is arguably the most fundamental concept of all, not only in physics, but in empirical science generally:  namely, the concept of measurement accuracy.  Measuring instruments are not born equal.  If one did not have a way to distinguish measurements which are in some sense ``good'' from measurements which are in some sense ``bad''---if one did not have what Busch \emph{et al}~\cite{Busch:2014a} call a ``figure of merit''---one would be forced to regard all measurements as being on the same footing.  There would, in fact, be no reason to prefer numbers obtained using a state-of-the-art photon counter from those obtained using the cheaper, less  demanding procedure of making a blind guess.  Under such conditions empirical science would be impossible.  Since physics has actually made huge advances over the last century  it is obvious that, on a practical level, experimentalists  have ways to distinguish good measurements from bad.  However, those   practical methods are not supported by an adequate understanding at the theoretical level.  

It is worth asking why, given the fundamental importance of the problem, progress has been so slow.   Although it is true that the problem is technically demanding, it appears to us that the main obstacle has always been, as it continues to be, conceptual.  The classical concept of error involves a comparison between the measured value and the true value, as it existed before the measurement was made.  The Bell-Kochen-Specker theorem~\cite{Bell:1966,Kochen:1967,Clifton:2000,Appleby:2005c}, however, requires us to abandon the idea that a measurement ascertains the pre-existing value of a specified observable.  This is such a radical departure from classical ideas that Bell~\cite{Bell:1982} suggested that ``the field would be significantly advanced by banning [the word `measurement'] altogether, in favour for example of the word `experiment' ''.   The question then arises:  once the classical concept of measurement has gone up in smoke what, if anything, is left of the classical concept of measurement accuracy?  It will be seen that this is  a special case of the more general question, which lies at the heart of all the disputes about quantum foundations:  once the classical concept of measurement has gone up in smoke what, if anything, is left of the classical concept of empirically accessible reality?  The problem is consequently of a rather peculiar kind.  Physics encompasses an enormous spectrum of problems, ranging from nuts-and-bolts problems such as measuring a length precisely, to deep philosophical questions.  The error-disturbance principle is unusual because it directly connects the two ends of the spectrum.  On the one hand it has, as we stressed above, an immediate, down-to-earth practical relevance.  On the other hand  we would argue that one of the factors obstructing progress---the reason almost half a century elapsed before people started to get seriously to grips with the problem---was the obscurities of the Copenhagen Interpretation.  It thus provides   a  riposte to the suggestion that the interpretational issues are  practically unimportant.  

Although the connections with the interpretation problem are not the main point of this paper, they are  part of the underlying motivation.  It is therefore appropriate to say something about them in this introductory section\footnote{We will give a  more detailed discussion in a subsequent publication}.
Let us begin by observing that Heisenberg himself did not propose, or even conjecture an error-disturbance principle.    He did, of course, construct his famous microscope argument~\cite{Heisenberg:1927,Heisenberg:1930}, which has suggested to many  that he had in mind such a principle.  However, that is based on a misunderstanding of the point of the microscope argument (that is, what \emph{Heisenberg} saw as the point).  What that point was emerges most clearly in Von Neumann's account~\cite{Neumann:1932}, where it is made completely explicit that the function of the microscope argument is   to give  intuitive support to the inequality  proved by Kennard~\cite{Kennard:1927} and Weyl~\cite{Weyl:1928} (the latter attributing the result to Pauli), and to its subsequent generalization by Robertson~\cite{Robertson:1929} and Schr{\"{o}}dinger~\cite{Schrodinger:1930}.  In his 1927 paper~\cite{Heisenberg:1927} Heisenberg was less explicit.  At the time he wrote the paper the Kennard-Pauli-Weyl proof was yet to come and, perhaps for that reason, he gave the microscope argument pride of  place.  However, he was  using it to support his original prototype for the uncertainty principle, namely the order of magnitude estimate $p_1 q_1 \sim h$, where $p_1$, $q_1$ are the standard deviations of the $p$, $q$ probability distributions scaled by a factor $\sqrt{2}$.  There is no indication that he envisaged, in addition to this statement, an entirely different error-disturbance principle.  

Nevertheless, although Heisenberg did not in fact propose  an error-disturbance  principle, one may  feel that he \emph{should} have done so, for  it is strongly suggested by the considerations in his 1927 paper (to that extent we agree with Busch \emph{et al}~\cite{Busch:2007}, that it is ``latent'' in what he says).  Reflecting on the microscope experiment it seems intuitively evident that the measurement of position really is (in some sense) less than perfectly accurate, and that the electron really will (in some sense) be disturbed by the photon.  The situation seems to be crying out for proper quantum mechanical analysis.  Yet it evidently did not seem that way to Heisenberg.  Nor, apparently, did it seem that way to most other people before the 1960s.  During the  period between 1927 and the 1965 paper of Arthurs and Kelly~\cite{Arthurs:1965}  one finds various paraphrases and elaborations of the  statements in Heisenberg's original paper but we are not aware of any clear statement of  the error-disturbance  principle conceived as a proposition distinct from the Kennard-Pauli-Weyl inequality,
or any recognition of the fact that a quantum mechanical definition of measurement accuracy is needed.  The question arises: Why is it that Heisenberg and so many others failed to draw what seems to most people now the obvious conclusion  from his uncertainty paper?  The answer, we suggest, is that their understanding was obstructed by one of the  features of the Copenhagen interpretation.  

In the words of Bell~\cite{Bell:1987a} the Copenhagen interpretation\footnote{Of course, the Copenhagen interpretation is a somewhat nebulous entity.  For one thing different proponents had different ideas (see Faye~\cite{Faye:2014} and references cited therein).  For another the views of individual proponents evolved in the course of time (see Plotnitsky~\cite{Plotnitsky:2013} for the evolution in Bohr's thinking, Camilleri~\cite{Camilleri:2009a} for the evolution in that of Heisenberg).   It is consequently impossible to give a characterization which is both concise and fully adequate.    However, it appears to us that Bell's one sentence summary does  identify a theme which, in one form or another, is common to all the variants.} divides the world ``into speakable apparatus \dots that we can talk about \dots and unspeakable quantum system that we can not talk about'' (ellipses in the original).  This idea has been hard to maintain since the 1970's, when it was realized, in connection with the problem of gravity-wave detection, that the error-disturbance principle is relevant to highly accurate measurements of a \emph{macroscopic} oscillator~\cite{Braginsky:1980,Braginsky:1992}.  Such an oscillator is just as speakable as any other piece of laboratory apparatus; yet at the same time we need to analyze its behaviour quantum mechanically.   But in the early days of quantum mechanics the unspeakability of quantum systems was accepted by almost everyone.  Thinking of the quantum world as  ineffable, and beyond the reach of thought~\cite{Plotnitsky:2013}---forgetting that the quantum world is the  one in front of our noses---encouraged the perception that quantum mechanical measurements are so utterly different from classical ones that no points of contact with classical concepts are possible.  In particular, it encouraged the assumption that the classical concept of error cannot carry over to quantum mechanics in any shape or form.   This, we would suggest, is why Heisenberg did not follow through on what now seems the obvious implication of his microscope argument, and formulate an error-disturbance principle.  He did not do so because he rejected the very idea of a quantum error, or a quantum disturbance.

Corresponding to the idea that there are two different worlds, speakable and unspeakable, there is a widespread assumption that there are two  kinds of measurement, classical and quantum. 
If highly accurate determinations of the center-of-mass motion of a macroscopic object are to be treated as quantum measurements then it is hard to see how one can consistently make such a distinction.  Instead, one seems forced to the view that every measurement is a quantum measurement, measurements with a meter rule not excluded.  To be sure, low precision measurements with a meter rule permit simplifying assumptions which cease to be valid as one increases the accuracy.  However, that is a purely a matter of practical convenience, not the signal of a fundamental difference of kind.  In the case of kinematics we continue to use the Newtonian theory when analyzing low velocity motion, without taking this to mean that there is a fundamental  difference of kind between the relativistic momentum of a space-ship travelling at near light speed and the Newtonian momentum of a train on the London underground.  Similarly in the case of measurements:  we need a unified description.  

In particular, we need a unified description of measurement errors.  The statement, that the kind of sophisticated measurement on a macroscopic object which demands a quantum  analysis is more accurate than a  commonplace  measurement  with a meter rule, tacitly assumes that there is a single concept of accuracy applicable to both. Otherwise, we would not have the basis for a comparison.  In the case of kinematics the Newtonian definition of momentum is an approximation to the relativistic definition, valid for low velocities.  In the same way, we need an overarching quantum definition of error, which effectively reduces to the classical one in limiting cases. At first sight this may seem impossible, since quantum mechanics requires us to drop the assumption that a measurement ascertains the pre-existing value of a specified observable.  However, on  further reflection it will be seen that even on classical assumptions one is never able  to directly compare the measured value with the pre-existing true one. In  classical physics as in  quantum physics, measured values are the only ones available.   It follows that, although in classical principle the error is the difference between the measured value and the  true one, in point of classical practice it must be possible to do everything using measured values only.

The purpose of this paper is to make a small beginning on the task of constructing a unified theory of measurement.  We focus on Busch, Lahti and Werner's (BLW's) criticisms~\cite{Busch:2014b,Busch:2014c,Busch:2014d,Busch:2014e} of the operator approach\cite{Appleby:1998b,Ozawa:2003,Ozawa:2003a,Ozawa:2013,Ozawa:2014} to the description of quantum errors and disturbances.  Their criticisms raise some issues which are highly relevant to the above discussion, and which need to be settled if we hope to make progress.  It should be stressed, that although our conclusion is that the operator approach is more useful than BLW allow, we are far from rejecting everything they say.  In particular, we completely agree with them on what is, perhaps, the most essential point, that quantum errors and disturbances need to be defined operationally.   Moreover, in defending the operator approach, it is no part of our intention to impugn the distributional approach they favour.   No one would say that the RMS characterization of an ordinary uncertainty is either ``better'' or ``worse'' than an entropic characterization.  Rather one has different quantitative measures each of which has advantages and disadvantages.  Similarly here.  The task is not to single out one particular approach as somehow canonical, but rather to achieve a clear understanding, at the basic conceptual level, of what is meant by the words ``error'' and ``disturbance'' in a quantum mechanical context, and of the different ways of quantifying the concepts.

There are  two  versions of the operator approach (or \OP approach as we will call it from now on).  BLW's criticisms are largely directed against the state-dependent version  proposed by Ozawa~\cite{Ozawa:2003,Ozawa:2003a}. However, we had previously proposed a state-independent version~\cite{Appleby:1998b}.  Both versions are relevant to our discussion.  In Section~\ref{sec:RMS} we compare and contrast them.

Section~\ref{sec:Justification} is the core of the paper.  We begin with the  classical concepts of error and disturbance.   We show that there at least two ways  to reformulate them in a manner which does not involve a comparison with pre-existing values.   We then show that the reformulated definitions have natural quantum generalizations, which we  call the  \DM and \CO definitions.   The \DM and \CO errors are thus candidates for the overarching concept of measurement accuracy which, we argued above, is necessary if one wants to construct a unified theory of measurement, in which every measurement is seen as quantum.  They also have an important bearing on BLW's criticism of the \OP approach.  As BLW correctly observe, the \OP definitions are non-operational.  However, the \DM and \CO definitions \emph{are} operational.  Moreover, the \OP quantities are upper bounds on the corresponding \DM and \CO quantities.  This gives indirect operational meaning to the \OP quantities.  Specifically, it means that if one of the \OP quantities is small, then there are at least two well-defined operational senses in which the measurement is accurate or non-disturbing.  The situation when an \OP quantity is large is more problematic.  In the state-independent case it is possible that smallness of the \OP error/disturbance is both necessary and sufficient for the measurement to be accurate/non-disturbing in a well-defined operational sense.  However, we have not been able to prove this.

In Section~\ref{sec:Comparison} we analyze BLW's  objections to the \OP approach in the light of the foregoing.  BLW contrast the operator approach with what they call a distributional approach.  It is to be observed, however, that the \DM and \CO quantities are also defined distributionally.   Since the \OP quantities owe their physical meaning to their connection with the \DM and \CO quantities, it follows that the \OP quantities  are indirectly distributional.  In short, the problem is not to decide between a distributional approach and some other, completely different approach. Rather it is to decide between two different kinds of distributional approach.  As with all such questions, the answer is relative to the situation of interest.  We show that there is at least one important class of physical problems for the which the \DM error, and by extension the \OP error, are clearly more appropriate than the definition which BLW favour, based on the Wasserstein 2-deviation.  

Finally, in the Appendix, we give a more careful proof of the error-disturbance and error-error relations than the one we presented in ref.~\cite{Appleby:1998b}.  In that earlier paper we skated over certain questions of domain and differentiability.  We here take the opportunity to fill in the missing details.

\section{The Operator Approach}
\label{sec:RMS}
In this section we outline the operator characterization of quantum errors and disturbances.  Our aim is purely descriptive.  We justify the approach, and respond to the various criticisms which have been made of it,  in subsequent sections.  

Consider a classical measurement of position.  Let $x_{\rm{i}}$, $p_{\rm{i}}$ be the position and momentum immediately before the measurement and let $x_{\rm{f}}$, $p_{\rm{f}}$ be their values immediately after it.  Let  $\mu_{\rm{f}}$ be the final value of the pointer observable.  Then the error in the measurement of position is $\mu_{\rm{f}}-x_{\rm{i}}$ and  the disturbance to the momentum is  $p_{\rm{f}} -p_{\rm{i}}$ (classical physics does not, of course, require there to be a disturbance to the momentum, but such a disturbance is perfectly possible).  On the level of formal analogy it is natural to ask what happens if one replaces the classical variables in these expressions with the corresponding Heisenberg picture  operators.   Let $\mcl{H}_{\rm{s}}$ and $\mcl{H}_{\rm{a}}$ be the Hilbert spaces for the system and apparatus respectively, and assume that system+apparatus are initially in the product state $\sy \otimes \ap$, where $\sy$ is density matrix of the system and $\ap$ is the density matrix of the apparatus. Let $\hat{U}\colon \mcl{H}_{\rm{s}} \otimes \mcl{H}_{\rm{a}} \to \mcl{H}_{\rm{s}} \otimes \mcl{H}_{\rm{a}}$ be the unitary operator describing the measurement interaction, let 
\ea{
\hat{x}_{\rm{i}} &= \hat{x}\otimes I, & \hat{p}_{\rm{i}} & = \hat{p} \otimes I, & \hat{\mu}_{\rm{i}} &=  I \otimes \hat{\mu},
}
 be the position, momentum and pointer Heisenberg picture observables immediately before the measurement interaction commences, and let
\ea{
\hat{x}_{\rm{f}} &= U^{\dagger} \hat{x}_{\rm{i}} U, & \hat{p}_{\rm{f}} &= U^{\dagger} \hat{p}_{\rm{i}} U, & \hat{\mu}_{\rm{f}} &= U^{\dagger} \hat{\mu}_{\rm{i}} U,
}
be the Heisenberg picture observables immediately after the interaction has finished.   Formal analogy with the classical case then suggests that we define\footnote{In  Appleby~\cite{Appleby:1998,Appleby:1998a,Appleby:1998b,Appleby:1999} we also introduced the predictive error operator $\hat{\mu_{\rm{f}}}-\hat{x}_{\rm{f}}$.  In this paper we will focus exclusively on the retrodictive operator since that is the one which gives rise to conceptual difficulties.}
\ea{
\eop{X} &= \hat{\mu}_{\rm{f}} - \hat{x}_{\rm{i}}  & \dop{P} &= \hat{p}_{\rm{f}} - \hat{p}_{\rm{i}}
}
We refer to $\eop{X}$ (respectively, $\dop{P}$) as the error (respectively, disturbance) operator.  We then obtain a numerical characterization of the error by defining
\ea{
\eva{x} &=\bigl(\Tr( \eop{X}^2 (\sy \otimes \ap) \bigr)^{\frac{1}{2}},
\\
\intertext{and a numerical characterization of the disturbance by defining}
\dva{p}  &=\bigl(\Tr( \dop{P}^2 (\sy \otimes \ap) \bigr)^{\frac{1}{2}}
}
We label the quantities with a superscript $\sy$ because, while the apparatus ``ready'' state  $\ap$ is assumed to be always the same,  the system state $\sy$ can vary.  
The  operators $\eop{X}$, $\dop{P}$ are  unbounded which means that the quantities $\eva{x}$, $\dva{p}$ are not defined for every state $\sy \otimes \ap$.  In the following we will always assume that $\sy$ is in the set of physical states $\mcl{P}$ defined in the Appendix.  If this is true then, provided that $\ap$ is appropriately chosen,  the expectation value $\Tr(M(\sy\otimes \ap)$ is well defined and finite  for every monomial $M$ in $\hat{x}_{\rm{i}}$, $\hat{x}_{\rm{f}}$, $\hat{p}_{\rm{i}}$, $\hat{p}_{\rm{f}}$, $\hat{\mu}_{\rm{i}}$,  $\hat{\mu}_{\rm{f}}$.

Of course, we have not yet justified the interpretation of $\eva{x}$ and $\dva{p}$  as an error and disturbance (beyond noting the  formal analogy with classical physics which, though suggestive, is clearly not sufficient to justify the proposal).    We defer  a proper justification to the next section and focus here on the question, whether there exists an error-disturbance relation expressible in terms of these quantities.  In various special cases\cite{Arthurs:1965,Arthurs:1988,Yuen:1982,Ishikawa:1991,Ozawa:1991,Appleby:1998,Appleby:1998a}
 one does indeed have
 \ea{
 \eva{x} \dva{p} \ge \frac{\hbar}{2}
 \label{eq:errDisNaive}
 }
analogous to the ordinary uncertainty relation $\Delta x \Delta p \ge \hbar/2$.  However, as we showed in ref.~\cite{Appleby:1998b},  it is easy to see that the inequality cannot be generally valid. Consider a simple model for the measurement process, in which the pointer observable $\hat{\mu}$ is the position of a particle having momentum $\hat{\pi}$ and in which   the measurement rotates the system particle position onto the pointer particle position, so that
\ea{
\hat{\mu}_{\rm{f}} = \hat{x}_{\rm{i}}.
}
Such a rotation is effected by
\ea{
\hat{U} &= e^{-\frac{i\pi \hat{H}}{2\hbar}}
\label{eq:cexUn}
}
where
\ea{
\hat{H} &= \hat{x}_{\rm{i}} \hat{\pi}_{\rm{i}} - \hat{\mu}_{\rm{i}} \hat{p}_{\rm{i}}.
}
(so if $\hat{x}$, $\hat{\mu}$ were different components of the position of a single particle in three dimensions $\hat{H}$ would be a component of the angular momentum operator).  The fact that $\hat{\mu}_{\rm{f}} = \hat{x}_{\rm{i}}$ means that  $\eop{X} = 0$.  It is  easy to see that $\dop{P} = -\hat{\pi}_{\rm{i}} - \hat{p}_{\rm{i}}$.  So this is a measurement for which the error is zero while the disturbance is finite for every physical state.  

Although we are mainly concerned with the error-disturbance relation in this paper it is worth noting that exactly the same argument shows~\cite{Appleby:1998b} that the error-error relation
\ea{
\eva{x} \eva{p} \ge \frac{\hbar}{2}
\label{eq:errerrSD}
}
for a joint measurement of position and momentum cannot be valid in general.  Indeed, consider a joint measurement in which the interaction of the particle with the position pointer is  described by the unitary in Eq.~\eqref{eq:cexUn}, while the momentum pointer just goes along for the ride, without interacting at all.  One then  has $\eva{x}=0$ and $\eva{p} = \sqrt{\la (\hat{\mu}_{\rm{P, i}}-\hat{p}_{\rm{i}})^2 \ra}$ (where $\hat{\mu}_{\rm{P,i}} = \hat{\mu}_{\rm{P,f}}$ is the momentum pointer position).  Even though the momentum is not really being measured at all, $\eva{p}$ is still finite for every physical state.  So Inequality~\eqref{eq:errerrSD} is violated for every physical state. 

The fact that Inequalities~\eqref{eq:errDisNaive},\eqref{eq:errerrSD} are not generally valid was noted by 
ourselves\footnote{In ref.~\cite{Appleby:1998b} we framed the discussion in the context of joint measurements of position and momentum.  However, the example we used to make the point was the model interaction described by Eq.~\eqref{eq:cexUn} above, in which the momentum pointer does not interact with the system at all, and which can therefore just as well be regarded as a single measurement of position only.  Generally speaking any measurement of position only can be regarded as a joint measurement in which the momentum pointer does not interact with the system.  Conversely, a joint measurement of both position and momentum becomes a single measurement of position only if we simply disregard the momentum reading.
}~\cite{Appleby:1998b} and subsequently by Ozawa~\cite{Ozawa:2002,Ozawa:2003,Ozawa:2003a,Ozawa:2003b,Ozawa:2004a}; in the case of \eqref{eq:errerrSD} also by Hall~\cite{Hall:2004}.    We, Ozawa, and Hall responded to these facts by trying to find  alternative inequalities which are generally valid.  However, we on the one hand, and Ozawa and Hall on the other,  were led in different directions.  We begin by describing our approach to the problem, since this came first in point of time.  

The essential point will emerge most clearly if we start with the  violation of Inequality~\eqref{eq:errerrSD} by the measurement described by Eq.~\eqref{eq:cexUn}.  For this measurement it is not simply that the product $\eva{x}\eva{p}$ is less than $\hbar/2$ for a certain subset of initial  states.  The product is in fact strictly zero for every possible initial state.  However, it would  be rash to conclude from this that the measurement is  in some sense ``best possible''.  As we noted above, the momentum pointer does not interact with the system, which means that so far as momentum is concerned the measurement is not only not highly accurate, it cannot properly be described as a measurement at all.  It is true that $\eva{p}$ is small for a certain, highly specific set of initial states.  However, that is not a reason for describing the measurement as accurate.   Consider the following scenario:
\begin{quote}
Alice goes to Bob's shop and buys what Bob says is a highly accurate ammeter.  However, when she gets home she finds that the needle is stuck at the $1$ amp position.  When she goes back to complain Bob is unrepentant.  He insists that the meter is indeed highly accurate provided one uses it to measure a $1$ amp current.
\end{quote}
Clearly, Alice will not be satisfied with this response.  No more would she be satisfied with the claim, that the interaction described by  Eq.~\eqref{eq:cexUn} gives a highly accurate measurement of momentum.

This example shows that the smallness of the product $\eva{x}\eva{p}$ is not always the signature of a highly accurate joint measurement of position and momentum.  Similar remarks apply to the product $\eva{x}\dva{p}$.   Consider, for instance, a ``measurement'' for which $\hat{U}$ is the identity, so that there is no coupling whatever between system and apparatus.  Here $\dva{p}$  is zero for every possible initial  state while $\eva{x}$ is always finite and sometimes small.  Yet, as in the broken ammeter example, it would be an abuse of language to describe this as a measurement of position which is always non-disturbing and sometimes highly accurate.

In ref.~\cite{Appleby:1998b} these considerations led us to look for replacements for the products $\eva{x}\eva{p}$, $\eva{x}\dva{p}$ whose smallness can unequivocally be regarded as the signature of a  measurement which is in some sense ``good''.  In the broken ammeter example what makes Bob's claim absurd is the fact that  an accurate classical ammeter is one for which the measured value is close to the true one, not just for one particular current, but for every current within a wide range.  Applying the same principle to the quantum case suggests that we define the error by
\ea{
\evb{x} &= \sup_{\sy\in \mcl{P}} \bigl(\eva{x}\bigr)
\label{eq:errXSupDefA}
}
where $\mcl{P}$ is the set of physical states, as defined in Appendix~\ref{sec:Technical}.  As we saw above, the smallness of $\eva{x}$ for some particular $\sy$ is consistent with the apparatus being completely decoupled from the system, so that it is not really measuring anything. But if $\evb{x}$ is small it means that $\eva{x}$ is small for every possible  state and we clearly are entitled to say that the measurement is  highly accurate (taking into account the discussion in Section~\ref{sec:Justification}).  Similar principles apply to the concept of disturbance.  Consider, for instance, the measurement described by Eq.~\eqref{eq:cexUn}, which rotates $\hat{\mu}$ onto $\hat{x}$.  For this measurement  $\dva{p}$ will be small for certain special choices of $\sy$ and $\ap$.  However, it will typically be large.  A medical procedure would not usually be described as non-invasive merely on the grounds that it can occasionally happen that the patient escapes almost intact.   Similarly here.  We accordingly define the disturbance to be
\ea{
\dvb{p} &= \sup_{\sy\in \mcl{P}} \bigl(\dva{p}\bigr).
\label{eq:errPSupDefB}
}
With these definitions it can be shown 
\ea{
\evb{x} \dvb{p} \ge \frac{\hbar}{2}
\label{eq:MyErrDisIdeal}
}
where we use the convention, here and elsewhere, that a product of the form $q \times \infty$ counts as infinite, even if $q=0$.
One can also prove a universally valid version of the error-error relation for a joint measurement of position and momentum
\ea{
\evb{x} \evb{p} \ge \frac{\hbar}{2}
\label{eq:MyErrErrIdeal}
}
where $\evb{p}$ is defined by taking the supremum of $\eva{x}$. 
In ref.~\cite{Appleby:1998b} we gave a proof of these relations which glossed over some questions to do with domains of definition and differentiability.  A completely rigorous proof is given in Appendix~\ref{sec:Technical} below.

The quantities $\evb{x}$, $\dvb{p}$ are not without interest, as we discuss below.  However, they are not the appropriate definitions for a real measuring instrument.  The demand that $\evb{x}$ be small is the demand that $\eva{x}$ be small, not only when $\sy$ is a wave-packet localized in the vicinity of the apparatus, but also when $\sy$ is a wave-packet localized on the other side of the cosmic event-horizon.  Clearly, this is  not a reasonable demand to make of a practical laboratory instrument, which is only designed to give accurate readings for a restricted set of input states.  In ref.~\cite{Appleby:1998b} we accordingly proposed the following modified definitions
\ea{
\evc{x} &= \sup_{\sy \in \mcl{R}} \bigl( \eva{x}\bigr),
\\
\dvc{p} &= \sup_{\sy \in \mcl{R}} \bigl( \eva{x}\bigr)
} 
where the supremum is now taken over a proper subset $\mcl{R}$ of the set of physical states.   We took $\mcl{R}$ to be a set of physical states for which the mean values $\Tr(\hat{x}\sy$, $\Tr(\hat{p}\sy)$ lie in a rectangular region of phase space with sides $\lX$, $\lP$, and satisfying certain additional conditions.
We then proved the inequalities
\ea{
\evc{x} \dvc{p} + \frac{\hbar}{\lX} \evc{x} + \frac{\hbar}{\lP}\dvc{p} \ge \frac{\hbar}{2},
\label{eq:MyErrDisReal}
\\
\evc{x} \evc{p} + \frac{\hbar}{\lX} \evc{x} + \frac{\hbar}{\lP}\evc{p} \ge \frac{\hbar}{2},
\label{eq:MyErrErrReal}
}
where we again use the convention that a product of the form $q \times \infty$ counts as infinite, even if $q=0$.
It will be observed that in the limit as $\lX$, $\lP \to \infty$ we recover Inequalities~\eqref{eq:MyErrDisIdeal}, \eqref{eq:MyErrErrIdeal}.  As with Inequalities~\eqref{eq:MyErrDisIdeal}, \eqref{eq:MyErrErrIdeal}, the proof of Inequalities~\eqref{eq:MyErrDisReal}, \eqref{eq:MyErrErrReal} which we gave in ref.~\cite{Appleby:1998b} glossed over certain   details.  We give a completely rigorous proof in Appendix~\ref{sec:Technical} below, where we also take the opportunity to strengthen the statement somewhat.

Let us now turn to the approach of Ozawa~\cite{Ozawa:2003,Ozawa:2003a,Ozawa:2004a,Ozawa:2003b} and Hall~\cite{Hall:2004}.  In our approach we replaced the state-dependent definitions $\eva{x}$, $\dva{p}$ with the quantities $\evb{x}$, $\dvb{p}$ and $\evc{x}$, $\dvc{p}$ and proved inequalities applying to those.  Ozawa~\cite{Ozawa:2003,Ozawa:2003a}, by contrast, kept with the state-dependent definitions and showed
\ea{
\eva{x}\dva{p} + \uca{p} \eva{x} + \uca{x} \dva{p} \ge \frac{\hbar}{2}
\label{eq:OzawaErrDis}
}
where $\uca{x}$, $\uca{p}$ are the ordinary  uncertainties in the state $\sy$.  He also showed that~\cite{Ozawa:2004a,Ozawa:2003b}, for a joint measurement of position and momentum, 
\ea{
\eva{x}\eva{p} + \uca{p} \eva{x} + \uca{x} \eva{p} \ge \frac{\hbar}{2}
\label{eq:OzawaErrErr}
}
It will be observed that these relations have a similar mathematical form to our~\eqref{eq:MyErrDisReal}, \eqref{eq:MyErrErrReal}.  Hall~\cite{Hall:2004} proved a relation similar to \eqref{eq:OzawaErrErr}.  Other modifications and improvements have also been proved~\cite{Weston:2013,Branciard:2014,Ozawa:2014a}.

The reader should not conclude from our earlier discussion that we have any  objection to the state-dependent definitions employed by Ozawa, Hall and others.  Asking whether a state-independent definition is better than a state-dependent one is like asking whether a hammer is better than a screw-driver.  The answer to all such questions, concerning the suitability of a tool, is relative to the use to which it is  put.  The fact that Bob, in the broken ammeter example, makes an inappropriate use of it does not invalidate the idea, that the classical error is the difference between the measured value and the true one.  Similarly here.  It is true that there exist quantum analogues of the broken ammeter---processes which do not properly count as a measurement for which the state-dependent error is small.  Nevertheless, the state-dependent error has a well-defined physical meaning (as we discuss in Section~\ref{sec:Justification}) and this makes it a potentially useful tool.   State-independent definitions, such as the ones proposed by ourselves or  BLW~\cite{Busch:2013,Busch:2014a}, have the advantage that they supply what BLW call an overall figure of merit; while state-dependent definitions, if not handled with care, can  lead to unreasonable conclusions.   But, as Rozema \emph{et al}~\cite{Rozema:2013} point out, state-independent definitions  have the  disadvantage that they are insensitive to  fine, state-dependent details which can be  important.  The state-dependent error can be used to analyze those details.  It is to be observed, furthermore, that  the state-dependent quantities $\eva{x}$, $\dva{p}$ are the limits of $\evc{x}$, $\dvc{p}$ as $\mcl{R}$ is shrunk to a single point.  If one takes the view that  use of $\eva{x}$, $\dva{p}$ is in all circumstances inappropriate then it is hard to see how one can avoid taking the view that use of $\evc{x}$, $\dvc{p}$ is also inappropriate when $\mcl{R}$ is very small.  Which raises the question:   ``Just how large has $\mcl{R}$ got to be in order for the use of  $\evc{x}$, $\dvc{p}$ to be justified?''  It is difficult to see how the answer can be other than arbitrary.     It appears to us that such discussions are fruitless, and that the solution to the quandary ``state-dependent or state-independent?'' is not to regard it as a quandary.  Instead of making a once-and-for-all choice we are free to use either or both, in a manner adjusted to the question of interest\footnote{See the end of Section~\ref{sec:Comparison} for further discussion of the points raised in this paragraph.}.

So far from being rivals Ozawa's inequalities and ours are closely related.
Let $\mcl{R}$ be any region satisfying the condition of Theorem~\ref{thm:mainresult} in the Appendix.   If we take the supremum on both sides of Ozawa's inequality~\eqref{eq:OzawaErrDis} we obtain the relation
\ea{
\evc{x}\dvc{p} + K_x \evc{x} +K_p \dvc{p} \ge \frac{\hbar}{2}
}
where
\ea{
K_x &= \sup_{\hat{\rho}\in\mcl{R}} (\uca{p}) & K_p &= \sup_{\hat{\rho}\in \mcl{R}}(\uca{x}).
}
This is weaker than our~\eqref{eq:MyErrDisReal} if $\lX$, $\lP$ are large but stronger if they, and the region $\mcl{R}$ are small.   It is probably fair to say that  that an experimenter will never be  committed to the proposition that the system state is \emph{precisely} $\sy$.  So what is in question is always a set of states $\mcl{R}$.  If $\mcl{R}$ is small  one will want to use Ozawa's inequalities, but if it is large one will want to use ours (provided $\mcl{R}$ satisfies the condition of Theorem~\ref{thm:mainresult}).

Although $\eva{x}$ will, in practice, only be small for a restricted set of states,  the limiting situation, when it becomes zero for all $\sy \in \mcl{P}$, is still conceptually important.  It can be shown (Appleby~\cite{Appleby:1999}, Ozawa~\cite{Ozawa:2004}, Busch~\cite{Busch:2004}) that the condition $\evb{x} =0$ is both necessary and sufficient for the distribution of measured values to be $\la x | \sy | x\ra$.  No real measuring instrument could have precisely this distribution of measured values for every input state $\sy$; in particular, it cannot do so for states such that the support of$\la x |\sy | x\ra$ is not compact\footnote{In this connection it may be worth remarking that the $x$ and $p$ space wave-functions cannot both have compact support.  So for every physical state $\sy$ at least one of the two distributions $\la x |\sy | x\ra$, $\la p |\sy |p\ra$ must be practically unrealizable.}.  Nevertheless the idea that $\la x |\sy | x\ra$ is the probability distribution for a measurement of position has played a fundamental role in physical thinking ever since Born~\cite{Born:1926} first proposed it\footnote{Strictly speaking Born only proposed that we interpret $\la \mbf{p} | \sy |\mbf{p}\ra$ as a probability distribution.}.  There is no problem here provided we understand the proposal to be, not that $\la x |\sy | x \ra$ is an operational distribution (one corresponding to an actual measurement), but that it is the canonical, or target distribution to which an operational  distribution may conform more or less well.  

A similar result can be proved for joint measurements minimizing the product $\evb{x}\evb{p}$:  Namely, that the product is minimized if and only if the distribution of measured values is the Husimi function (Appleby~\cite{Appleby:1999}, Werner~\cite{Werner:2004}, Busch~\emph{et al}~\cite{Busch:2007}).  In Appleby~\cite{Appleby:2000a} we extended the analysis to measurements of angular momentum, and showed that a determination of spin-direction is optimal if and only if the  distribution of measured values is $\la \mbf{n} | \sy |\mbf{n}\ra$, where $|\mbf{n}\ra$ is  a suitably normalized $\SU(2)$ coherent state.

\section{Physical Interpretation of the Operator Definitions}
\label{sec:Justification}
We now come to the problem of interpreting the quantities defined in the last section. 
Quantum mechanics forces us to drop the classical assumption that a measurement ascertains the pre-existing value of a specified observable~\cite{Bell:1966,Kochen:1967,Clifton:2000,Appleby:2005c}.  Even if one postulates that the observable measured does have a pre-existing value, that  value must typically differ from the value found by measurement.  In the Bohm theory, for example, the result of a measurement of velocity is usually quite different from the postulated pre-existing velocity\cite{Bohm:1993,Holland:1993,Appleby:1999c,Appleby:1999d}.  Classically, the error is usually defined in terms of the difference between the measured value and the pre-existing true one.  It might consequently seem that, in abandoning the idea that measurements ascertain pre-existing values, we are obliged also to abandon the concept of experimental error (in the Introduction we argued that that is exactly how it did seem to, for example, Heisenberg).   We begin by showing that that is not the case.  Specifically we describe a classical model for which the classical error can be defined in a way that does not involve a comparison with pre-existing values.  We then show that this alternative definition naturally carries over to quantum mechanics..

The example we consider is that of a one-dimensional classical gas.  Let $x$ and $p$ be the position and momentum of a  particular particle in this gas,  and let $\syc(dx dp)$ be the phase space probability measure.  Suppose we  measure $x$.  Let $\mu_{\rm{f}}$ be the pointer position after the measurement.  We assume that the measurement process is stochastic and is described by a transition kernel\footnote{$\chi(A| x,p)$ is not in general a conditional probability because the definition $P(A|B) = P(A\cap B) /P(B)$ breaks down if $P(B)=0$.  However, it becomes a conditional probability if one assumes there is non-zero probability of the initial position and momentum being precisely $x$, $p$.} $\chi(d\mu_{\rm{f}} | x,p)$ such that the expectation value of a function $f(x,p,\mu_{\rm{f}})$  is  given by (see, for example, Cinlar~\cite{Cinlar:2011}) 
\ea{
\la f \ra^{\syc} = \int \left( \int  f(x,p,\mu_{\rm{f}}) \chi(d\mu_{\rm{f}} | x,p)\right) \syc(dx dp).
}
The superscript $\syc$ is to serve as a reminder that $\syc$ is  arbitrary, unlike $\chi$ which characterizes the measurement interaction and is therefore fixed. Define\footnote{$\sigma_{\rm{ce}}(x,p)$ is the RMS difference between the measured value and the pre-existing true one if one makes the physically unrealistic assumption that the initial position and momentum are known to be precisely $x$ and $p$.}
\ea{
\sigma_{\rm{ce}}(x,p) &= \left( \int (\mu_{\rm{f}}-x)^2 \chi(d\mu_{\rm{f}}|x,p)\right)^{\frac{1}{2}}.
}
It will be seen that $\sigma_{\rm{ce}}(x,p)$ is the RMS difference between the measured value and the pre-existing true one when $\syc$ is concentrated on the single point $(x,p)$.  
 We then define the classical error by
\ea{
\Delta_{\rm{ce}} x = \sup_{x,p\in \fd{R}} \left(\sigma_{\rm{ce}}(x,p)\right).
\label{eq:classErrDefA}
}
Of course, this definition is open to the same objection as the quantity $\evb{x}$ defined in the last section; namely, that it is likely to be infinite for a realistic model.  However, this need not detain us because we are not interested in the model for its own sake, but only as a conceptual bridge which will take us from classical intuition to a reasonable quantum mechanical definition of measurement error.  Now let 
\ea{
\bar{x}^{\syc} &= \int x  \syc(dx dp) , & \Delta^{\syc} x &=\left( \int (x-\bar{x}^{\syc})^2 \syc(dx dp)\right)^{\frac{1}{2}}
}
 be the mean and standard deviation relative to the measure $\syc$.  Then
\ea{
\sqrt{\la (\mu_{\rm{f}}- \bar{x}^{\syc})^2 \ra^{\syc}} &=\left( \la(\mu_{\rm{f}}-x)^2 \ra^{\syc} + \la (x-\bar{x}^{\syc})^2\ra^{\syc} - 2 \la (x-\bar{x}^{\syc})(\mu_{\rm{f}}-x) \ra^{\syc}  \right)^{\frac{1}{2}}
\nn
&\le \Delta^{\syc} x + \sqrt{\la (\mu_{\rm{f}}-x)^2 \ra^{\syc} }
\nn
&\le \Delta^{\syc} x + \Delta_{\rm{ce}} x.
}
Note that $\sqrt{\la (\mu_{\rm{f}}- \bar{x})^2 \ra^{\syc}}$, $\Delta^{\syc} x$ are $\syc$-dependent, but $\Delta_{\rm{ce}} x$ is not.  The inequality is actually tight.  To see this  choose a sequence $(x_n,p_n)$ such that $\sigma_{\rm{ce}}(x_n,p_n)  \to \Delta_{\rm{ce}} x$, and let $\syc_n$ be the measure concentrated on the point $(x_n,p_n)$.  Then
\ea{
\sqrt{\la (\mu_{\rm{f}}- \bar{x}^{\syc_n})^2 \ra^{\syc_n}} - \Delta^{\syc_n} x &\to \Delta_{\rm{ce}} x.
}
So
\ea{
\Delta_{\rm{ce}} x &= \sup_{\syc \in \mcl{M}} \left( \sqrt{\la (\mu_{\rm{f}}- \bar{x}^{\syc})^2 \ra^{\syc}} - \Delta^{\syc} x \right),
\label{eq:classErrDefB}
}
where $\mcl{M}$ is the set of all phase-space probability measures.  This gives us an alternative formula for the classical error.
We can derive a similar formula for the classical disturbance.  Let $p_{\rm{f}}$ be the momentum after the measurement and $\xi(dp_{\rm{f}} | x,p)$ the transition kernel such that
\ea{
\la f\ra^{\syc} &= \int \left( \int f(p_{\rm{f}},x,p) \xi(d p_{\rm{f}}|x,p)\right) \syc(dxdp).
}
Define the classical disturbance by
\ea{
\Delta_{\rm{cd}} p &= \sup_{x,p\in \fd{R}} \left( \int (p_{\rm{f}}-p)^2 \xi(dp_{\rm{f}}| x,p)\right)^{\frac{1}{2}}.
\label{eq:classDisDefA}
}
Then, by an argument similar to the one above, we find
\ea{
\Delta_{\rm{cd}} p &=\sup_{\syc \in \mcl{M}}\left( \sqrt{\la (p_{\rm{f}} - \bar{p}^{\sy})^2\ra^{\syc}} - \Delta^{\syc} p\right).
\label{eq:classDisDefB}
}
We are now free to throw away the ladder and take Eqs.~\eqref{eq:classErrDefB}, \eqref{eq:classDisDefB} to be the \emph{definitions} of $\Delta_{\rm{ce}} x$, $\Delta_{\rm{cd}} p$.  These alternative definitions do not involve a direct comparison between the measured value and the pre-existing one.  Consequently, they do not involve the expectation values of products of pairs of variables like $\mu_{\rm{f}}$ and $x$ which, in a quantum mechanical context, become non-commuting operators.  Instead they are framed in terms of the moments of probability distributions which are also defined in quantum mechanics.  They therefore generalize.  Just as we can classically, so in quantum mechanics, we can define the error and disturbance in terms of the increase in an RMS deviation from an initial state mean:
\ea{
\ebvb{x} &=\sup_{\sy \in \mcl{P}} \left( \ebva{x} \right), & \ebva{x} &= \sqrt{\Tr\bigl( (\hat{\mu}_{\rm{f}} -\bar{x}^{\sy})^2 (\sy\otimes \ap)\bigr) }-\uca{x},
\label{eq:ebvbDef}
\\
\dbvb{p} &= \sup_{\sy \in \mcl{P}} \left(\dbva{p} \right), & \dbva{p} &= \sqrt{\Tr \bigl(\hat{p}_{\rm{f}} -\bar{p}^{\sy})^2 (\sy \otimes \ap)\bigr) }-\uca{p},
\label{eq:dbvbDef}
}
where we employ the notations of the last section, together with
\ea{
\bar{x}^{\sy} &= \Tr(\hat{x} \hat{\rho}) , & \bar{p}^{\sy} &= \Tr(\hat{p}\hat{\rho}).
}
We may also define
\ea{
\ebvc{x} &=\sup_{\sy \in \mcl{R}} \left( \ebva{x} \right), & \dbvc{p} & = \sup_{\sy \in \mcl{R}} \left( \dbva{p} \right)
}
We will refer to these as the \DM definitions (``D'' for ``maximal increase in the RMS \textbf{d}eviation from the initial state mean'').  They are important because they  show that the Bell-Kochen-Specker theorem is not, as it might  seem, an insuperable obstacle blocking the path from the original classical intuition to a satisfactory quantum generalization.  On the contrary, if the concepts are appropriately formulated there is complete continuity between  classical and quantum  in this regard.  However, although the \DM definitions are valid and useful, they should not be regarded as canonical.  In the first place, there are  other classical definitions which also have natural quantum generalizations (as we will see in the next paragraph).  In the second place,  there is no reason to make classical physics  the  arbiter.  There may be useful quantum definitions which are not the generalization of any classical concept.

We arrive at another natural generalization of  classical ideas if we consider  measurements on a pair of correlated particles.   Suppose we have two particles with positions $\hat{x}_{A}$, $\hat{x}_{B}$ and momenta $\hat{p}_A$, $\hat{p}_B$, and suppose we measure $\hat{x}_{B}$.  Suppose that the unitary operator describing the measurement interaction is of the form $\hat{I} \otimes \hat{U}$, where $\hat{U}$ acts  on $\mcl{H}_{B}\otimes \mcl{H}_{\rm{ap}}$ in the product $\mcl{H}_A \otimes \mcl{H}_B \otimes \mcl{H}_{\rm{ap}}$ ($\mcl{H}_A$, $\mcl{H}_B$, $\mcl{H}_{\rm{ap}}$  being respectively the Hilbert spaces of particles $A$, $B$ and the apparatus).  Let $\hat{\mu}_{B}$ be the pointer position.  Classically, it would be natural to define the error to be the maximal increase 
in the correlation $\la (\mu_{B,\rm{f}} - x_{A})^2 \ra$ as compared to $\la ( x_{B,\rm{i}} - x_{A})^2\ra$, and the disturbance to be the maximal increase in the correlation $\la (p_{B,\rm{f}} - p_A)^2\ra$ as compared to $\la (p_{B,\rm{i}} -p_A)^2\ra$.  The point to notice here is that $[\hat{\mu}_{B,\rm{f}} ,\hat{x}_{A}] = [\hat{x}_{B,\rm{i}},\hat{x}_{A}] =[\hat{p}_{B,\rm{f}},\hat{p}_A] = [\hat{p}_{B,\rm{i}}, \hat{p}_A] =0 $.  So the classical definitions are expressed in terms of the moments of probability distributions which are also defined quantum mechanically.  They therefore generalize to
\ea{
\etvb{x} &= \sup_{\sy \in \mcl{P}} \left( \etva{x}\right), & \dtvb{p} &= \sup_{\sy \in \mcl{P}} \left( \dtva{x}\right),
}
where
\ea{
\etva{x} 
&=\sup_{\syj \in \spj} \left(\sqrt{\Tr\bigl( (\hat{\mu}_{B,\rm{f}} - \hat{x}_{A})^2 (\syj \otimes \ap)\bigr)}  \right.
\nn
& \hspace{2 in} - \left. \sqrt{\Tr(\bigl( \hat{x}_{B,\rm{i}} - \hat{x}_{A})^2 (\syj \otimes \ap)\bigr)}\right),
\label{eq:CeDe}
\\
\dtva{p}
&= \sup_{\syj \in \spj} \left(\sqrt{\Tr\bigl( (\hat{p}_{B,\rm{f}} - \hat{p}_{A})^2 (\syj \otimes \ap)\bigr)} \right.
\nn
&\hspace{2 in} \left. - \sqrt{\Tr(\bigl( ( \hat{p}_{B,\rm{i}} - \hat{p}_{A})^2 (\syj \otimes \ap)\bigr)}\right),
\label{eq:CdDd}
}
 $\spj$ being the set of physical states $\hat{\rho}_{AB}$ such that $\Tr_{A} (\hat{\rho}_{AB}) = \hat{\rho}$. We may also define
\ea{
\etvc{x} &=\sup_{\sy \in \mcl{R}} \left( \etva{x} \right), & \dtvc{p} & = \sup_{\sy \in \mcl{R}} \left( \dtva{p} \right),
}
We refer to these quantities as the \CO definitions (``C'' for ``Correlation'').  

Let us now turn to the definitions in Section~\ref{sec:RMS}, which we will refer to as the \OP definitions (``O'' for ``Operator'').   The commutators $[\hat{\mu}_{\rm{f}},\hat{x}_{\rm{i}}]$ and $[\hat{p}_{\rm{f}},\hat{p}_{\rm{i}}]$ are typically non-zero so the \OP quantities are typically not generalizations of the corresponding classical quantities, as Busch and co-workers have stressed\cite{Busch:2004,Busch:2014b}. The \OP quantities do, however, impose bounds on the \DM and \CO quantities, and this gives them an indirect physical interpretation.  We have
\ea{
 \eva{x}
&= 
\biggl(\Tr\Bigl( \bigl((\hat{\mu}_{\rm{f}} - \bar{x}^{\sy})^2 - \{ (\hat{\mu}_{\rm{f}} - \bar{x}^{\sy}), (\hat{x}_{\rm{i}} - \bar{x}^{\sy})\} + (\hat{x}_{\rm{i}} - \bar{x}^{\sy})^2 \bigr)\bigl(\sy \otimes \ap\bigr) \Bigr) \biggr)^{\frac{1}{2}}
}
from which it follows
\ea{
\ebva{x} &\le \eva{x} \le \ebva{x} + 2 \uca{x}.
\label{eq:dcoBndA}
}
for all $\sy \in \mcl{P}$.
Similarly 
\ea{
\dbva{x} &\le \dva{x} \le \dbva{x} + 2 \uca{x}.
\\
\etva{x} &\le \eva{x} \le \etva{x} + 2 \uca{x}
\\
\dtva{x} &\le \dva{x} \le \dtva{x} + 2 \uca{x}
}Taking suprema we deduce
\ea{
\ebvc{x} & \le \evc{x} \le \ebvc{x}+2\sup_{\sy \in \mcl{R}} \left(\uca{x}\right),
\\
\dbvc{x} & \le \dvc{x} \le \dbvc{x}+2\sup_{\sy \in \mcl{R}} \left(\uca{x}\right),
\\
\etvc{x} & \le \evc{x} \le \etvc{x}+2\sup_{\sy \in \mcl{R}} \left(\uca{x}\right),
\\
\dtvc{x} & \le \dvc{x} \le \dtvc{x}+2\sup_{\sy \in \mcl{R}} \left(\uca{x}\right).
}
When $\mcl{R} = \mcl{P}$ these inequalities reduce to
\ea{
\ebvb{x} &\le \evb{x}, & \dbvb{x} &\le \dvb{x}, & \etvb{x} &\le \evb{x}, & \dtvb{x} &\le \dvb{x}.
\label{eq:oboundcd}
}
We also have the following constraints on the relative sizes of the \DM and \CO quantities
\ea{
\left| \etva{x}-\ebva{x} \right| &\le 2\uca{x},  &  \left| \dtva{x}-\dbva{x} \right| &\le 2\uca{x}, 
\\
\left| \etvc{x}-\ebvc{x} \right| &\le 2\sup_{\sy \in \mcl{R}} \left(\uca{x}\right),  &  \left| \dtvc{x}-\dbvc{x} \right| &\le 2\sup_{\sy \in \mcl{R}} \left(\uca{x}\right).
\label{eq:CDRcompare}
}
These inequalities mean, among other things, that the \OP quantities are upper bounds on the corresponding \DM and \CO quantities.   

Our discussion raises some important questions.   If $\sup_{\sy \in \mcl{R}} \left(\uca{x}\right)$ is large then the above inequalities are consistent with one of the \OP quantities being large while the corresponding \DM and \CO quantities are both small.  
They also leave open the possibility that, in the case when $\sup_{\sy \in \mcl{R}} \left(\uca{x}\right)$ is large, one of the \DM quantities is large  while the corresponding \CO quantity is small, or \emph{vice versa}. 
One would like to know if these possibilities are actually realized.

Korzekwa \emph{et al}~\cite{Korzekwa:2014}  answer the first of these questions for the case of the state-dependent disturbances.  Consider two non-commuting observables $\hat{R}$, $\hat{S}$ on a finite dimensional Hilbert space.  Suppose that the system is initially in an eigenstate of $\hat{R}$ which is not also an eigenstate of $\hat{S}$, and suppose that one makes a von Neumann measurement of $\hat{R}$.  Then  the \DM and \CO disturbances are both zero  while the \OP disturbance is non-zero.

Busch~\cite{Busch:2014h} gives an example which shows that it is possible for the state-dependent  \DM and \CO errors to be zero while the state-dependent \OP error is non-zero.  Unlike Korzekwa \emph{et al}'s example it is rather artificial (it is a quantum version of the broken-ammeter scenario); however, it is enough to establish the point of principle.  Suppose the system and pointer particles are both spin-1/2 particles, and that the measured observable and pointer observables are the $\hat{\sigma}_z$ operators for their respective particles.  Suppose that the initial system+apparatus state is $|\psi\ra \otimes |\psi\ra$, and that $\hat{U} = \hat{I}$.  Then it is easily seen that the state-dependent \OP error is the ordinary uncertainty of $\hat{\sigma}_z$ in the state $|\psi\ra$, while the state-dependent \DM and \CO errors are zero.

We can use a modification of this example to show that it is possible for the state-dependent \DM quantities to be zero while the state-dependent \CO quantities are non-zero.  Let everything be as in the last paragraph except that system+apparatus are in the maximally mixed state $(1/4) \hat{I} \otimes \hat{I}$.  Then the \DM error is zero while the \CO error is $\sqrt{2}$ (the supremum in Eq.~\eqref{eq:CeDe} being achieved for the maximally entangled state $\hat{\rho}_{AB} = |\Psi\ra\la \Psi|$ with $|\Psi\ra = (1/\sqrt{2}(|+\ra \otimes |+\ra + |-\ra \otimes |-\ra)$, where $|\pm\ra$ are the eigenstates of $\hat{\sigma}_z$).  
To show that the same is true of the \DM and \CO disturbances continue to assume that system+apparatus are  in the maximally mixed state, but take the evolution operator $\hat{U}$ to be $\hat{\sigma}_y \otimes \hat{I}$.  Then the state-dependent \DM disturbance to the observable $\hat{\sigma}_x$ is zero, while the state-dependent \CO disturbance is $\sqrt{2}$ (the supremum in Eq.~\eqref{eq:CdDd} being achieved for the maximally entangled state $\hat{\rho}_{AB} = |\Psi\ra\la \Psi|$ with $|\Psi\ra = (1/\sqrt{2}(|+\ra \otimes |+\ra + |-\ra \otimes |-\ra)$, where $|\pm\ra$ are the eigenstates of $\hat{\sigma}_x$).

Of course, the last three examples (unlike the example of Korzekwa \emph{et al}) are somewhat artificial.  It would be interesting to see if the conclusion continues to hold for more realistic measuring processes.  Also, we have not addressed the  more challenging, and to our mind more interesting question, of what can be said in the state-\emph{independent} case.  This requires further investigation.

The \DM and \CO quantities have a direct, operational interpretation as errors and disturbances.  Smallness of one of these quantities is both necessary and sufficient for the measurement to be accurate or non-disturbing in a well-defined, operational sense.  By contrast the interpretation of the \OP quantities, as we have presented it here, is indirect:  their meaning comes from the fact that they supply various bounds on the \DM and \CO quantities.  Moreover, although smallness of an \OP quantity is sufficient, we have not been able to show that it is necessary for the measurement to be accurate or non-disturbing in a well-defined sense.  
In the case of the state-independent quantities it is possible that, with more work, one could establish necessity as well.  If that were so it would mean, in effect, that the state-independent \OP quantities were fully operational characterizations of the error and disturbance.  

Finally, let us note that there is no reason to assume that our analysis is complete.  The \OP quantities may capture other operationally identifiable features of the measurement which the \DM and \CO quantities both miss.

\section{Response to criticisms}
\label{sec:Comparison}
We now consider BLW's critique of the \OP definitions (also see Busch \emph{et al}~\cite{Busch:2004} and Korzekwa \emph{et al}~\cite{Korzekwa:2014}).    BLW contrast the \OP approach with what they call a distributional approach.  They argue that, although the \OP approach has its uses in certain special cases, the version of the distributional approach based on the Wasserstein 2-deviation is, in general, greatly preferable.  In addressing their criticisms let us begin by observing that the \DM and \CO definitions are themselves   distributional definitions.  Moreover,  although the \OP quantities are not  defined distributionally, their physical interpretation  (as given in Section~\ref{sec:Justification})  depends on the fact that they supply various  bounds on the corresponding \DM and \CO quantities.  So the distinction between operator and distributional approaches is less clear-cut than may initially appear.  The problem is not really to decide between a distributional approach and some other completely different approach; rather it is to decide between  alternative versions of the distributional approach.  As with all such problems the answer is  dependent on the  situation of interest.  In the following it is certainly not our intention to suggest that the \OP definitions are preferable to BLW's definitions  in every situation.  We only argue that there is a physically important class of situations in which the \DM definitions, and consequently the \OP definitions, are preferable.  

BLW accept that the \OP definitions give valid characterizations of the error (respectively disturbance) under conditions where the observables $\hat{x}_{\rm{i}}$, $\hat{\mu}_{\rm{f}}$ (respectively $\hat{p}_{\rm{i}}$, $\hat{p}_{\rm{f}}$) commute.  However, in cases where these observables do not commute they argue that $\hat{x}_{\rm{i}}$, $\hat{\mu}_{\rm{f}}$ (respectively $\hat{p}_{\rm{i}}$, $\hat{p}_{\rm{f}}$)  are not jointly measurable and, consequently, that the interpretation of $\eop{X}$, $\dop{P}$ as error and disturbance operators is ungrounded.  This objection would be justified if we were relying on a naive, purely formal  analogy with the classical expressions $\la (\mu_{\rm{f}} - x_{\rm{i}})^2 \ra$ and $\la (p_{\rm{f}} - p_{\rm{i}})^2\ra$.  However, since we are actually relying on the fact that the \OP quantities bound the \DM and \CO quantities, and since the definitions of the latter are just as operational as BLW's own definitions, there is no problem here.

BLW go on to substantiate their criticisms by giving examples of measurements where the \OP error is zero even though the distribution of measured values is quite different from the initial state distribution.  We will here confine ourselves to their Example 7.  The reader will easily perceive that a suitably modified version of our discussion  applies to their Examples 8, 9 and 10 (also to Example 5 in ref.~\cite{Busch:2004}). The example is of a measurement of position in which the POVM describing the distribution of measured values is the spectral measure of the shifted oscillator Hamiltonian
\ea{
\hat{H}' &= \hat{x} + \alpha\left( \hat{H}-\frac{\hbar\omega}{2}\right), &  \hat{H} &= \frac{1}{2m} \hat{p}^2 +\frac{m\omega^2}{2} \hat{x}^2
}
and in which the initial system state is the ground state of $\hat{H}$.  It is easily verified that $ \eva{x} =0$.  On the other hand it can be seen from Fig.~\ref{fig1} 
\begin{figure}[htb]
\includegraphics[width=4 in]{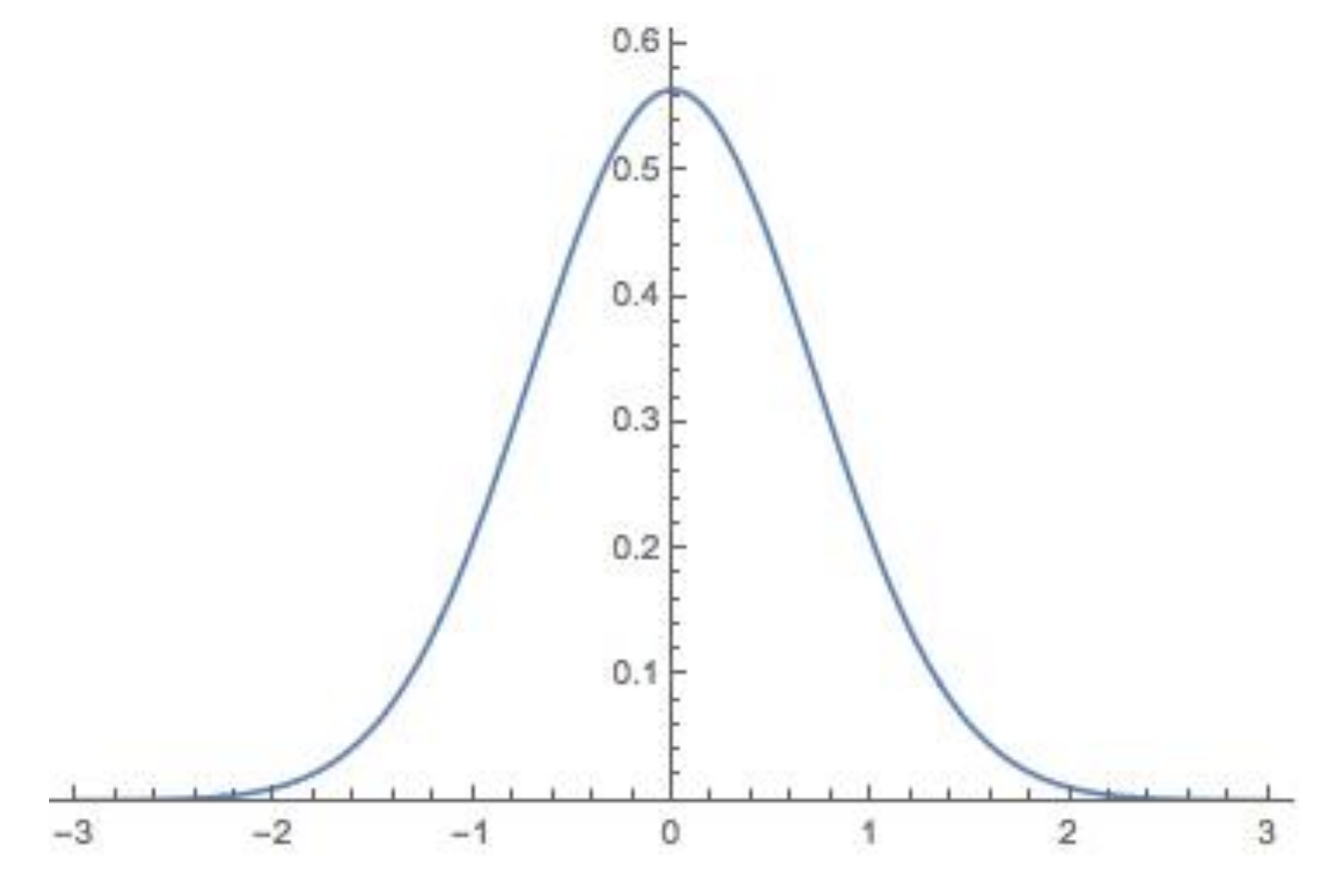}

\vspace{0.2 in}

\includegraphics[width=4 in]{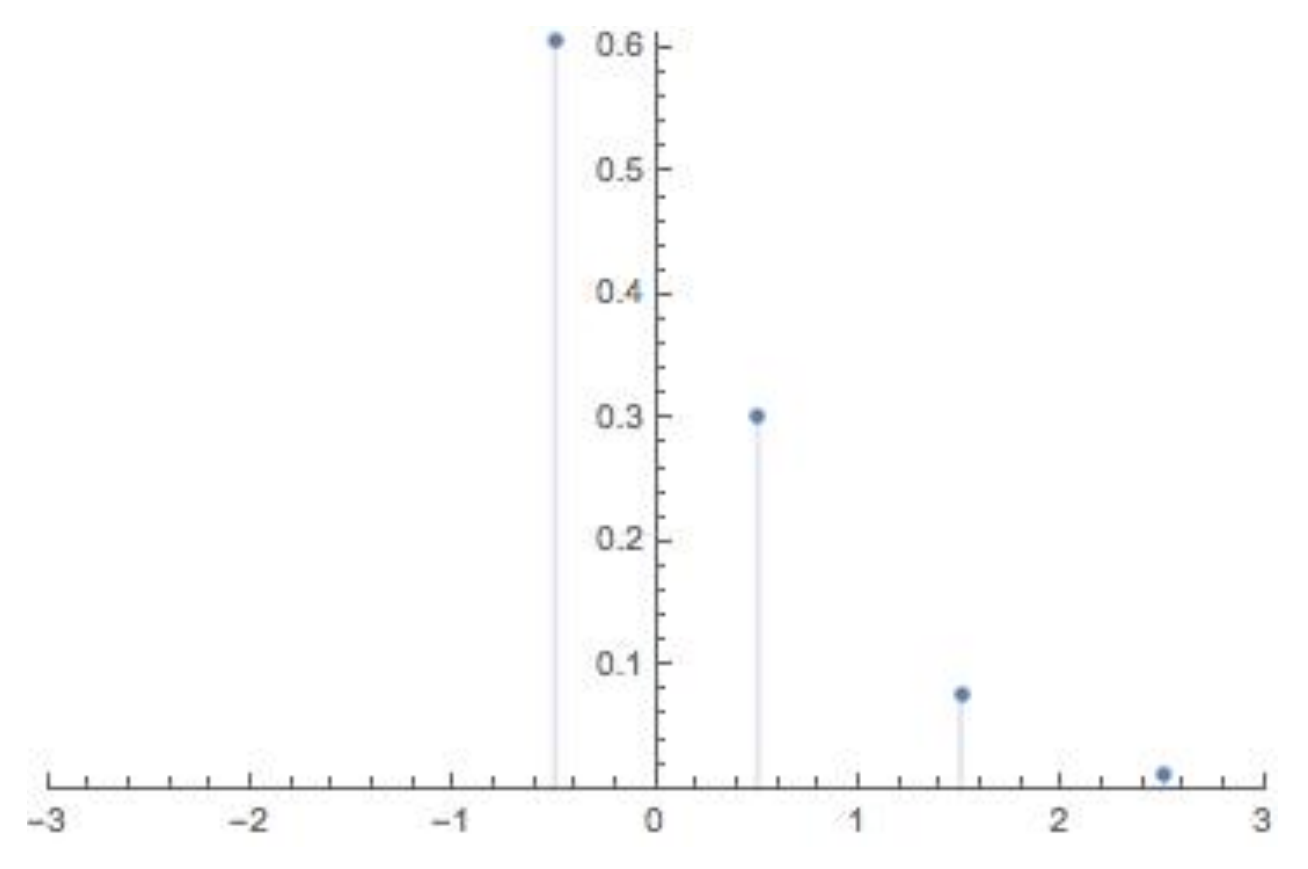}
\caption{Probability distributions for $\hat{x}_{\rm{i}}$ (top graph) and $\hat{\mu}_{\rm{f}}$ (bottom graph) in BLW's Example 7 (units chosen so that $\hbar = m/\alpha = \alpha \omega = 1$).\label{fig1}}
\end{figure}
that the probability distributions for $\hat{x}_{\rm{i}}$ and $\hat{\mu}_{\rm{f}}$ are very different.  In particular the distribution for $\hat{x}_{\rm{i}}$ is continuous whereas that for $\hat{\mu}_{\rm{f}}$ is discrete.  BLW take this to mean that the measurement is highly inaccurate, and that the \OP definition of error is correspondingly misleading.  They are right to the extent that there are  applications---tomography, for example---for which this measurement would be very ill-suited.  However, the purpose of a measurement is not always to accurately reproduce the initial state probability distribution.
That is obviously the case in classical physics.  Consider, for instance, measurements using a digital ammeter.  Here too  the initial state probability distribution  is continuous while the distribution of measured values is discrete.  But this would not usually be seen as a reason for preferring an analogue meter.  Similarly in quantum physics: There are  situations where one is only concerned with certain specific features of the distribution of measured values,  its detailed shape being otherwise unimportant.   Consider, for instance, a state discrimination problem where Bob is promised one of a finite set of Gaussian wave-packets and  has to decide which particular wave-packet Alice has sent.  In this case the crucial requirement is that the quantity $\la (\hat{\mu}_{\rm{f}} - \bar{x})^2 \ra$  be as small as possible.  Other considerations, such as the difference between continuous and discrete, are  irrelevant.   Consequently, a measurement like the one described in BLW's Example 7 is very well-suited to Bob's purpose.  The distributions depicted in Fig.~\ref{fig1} are indeed very different.  However, they have exactly the same mean and variance. Consequently, $\la (\hat{\mu}_{\rm{f}} - \bar{x})^2 \ra$ is not enlarged at all as compared to the initial state variance.  This is what Bob requires.  It is also (see Inequalities~\eqref{eq:oboundcd}) one of the pieces of information conveyed by the statement that $ \eva{x} =0$, which is not misleading at all---\emph{provided} it is correctly understood.  By contrast, the Wasserstein 2-deviation would be decidedly misleading if applied to this situation, as it would cause one to prefer, to the measurement described, one for which the second distribution was a smeared out version of the first---even though this would clearly be worse for Bob's particular purposes.    

Similarly with the disturbance:  in a situation where one is  interested in  the deviation from the initial state mean, but not in any other feature of the probability distribution, then the \DM definition, and consequently the \OP definition of disturbance will be more useful than the one based on the Wasserstein 2-deviation.

It is seldom, if ever, the case, that a single figure of merit captures every potentially relevant feature of a  piece of technology.  Suppose one is buying a car.   If one wants a vehicle that can drive very fast round a carefully prepared track one will choose one figure of merit; if, on the other hand, one wants a vehicle suitable for conveying a family of six to the beach one will choose another, quite different figure of merit.  Similarly with quantum measurements. 

In their examples 7--10 BLW criticize the \OP definitions on the grounds that the \OP error can be zero in situations where the initial state and final pointer  distributions are very different.  In examples 4 and 6 of Busch \emph{et al}\cite{Busch:2004} and example 3 of Busch~\cite{Busch:2014h} the authors make the opposite point,  that the \OP error can be large in situations where the initial state and final pointer distributions are identical; a fact which they regard as an evident defect of the operator approach.  Their argument is based on the principle, that a perfectly accurate measurement is one which perfectly reproduces the initial state probability distribution.  To see that the principle is not generally valid consider the following scenario:
\begin{quotation}
Alice lives in a city where $50\%$ of the population are infected with HIV.  She is worried that she may have it, so she goes to her doctor Bob to be tested.  Bob pulls a coin out of his pocket and tosses it.  He then puts on a grave face and says ``I am sorry, I have bad news for you.''  Alice is outraged, on the grounds that this isn't a proper test.  Bob, however, insists that it is a proper test.  After all, it has the same probability distribution.  What more can she want?
\end{quotation}
This is  a classical example.  One can easily construct a quantum example.  Suppose, for instance, that Alice and Bob are two students who want to perform a test of the Bell inequalities.  Unfortunately they cannot afford state-of-the-art photon counters so they decide that Alice will toss a fair coin at her station, and Bob will independently toss another fair coin at his.  On the principle adopted in examples 4 and 6 of ref.~\cite{Busch:2004} and example 3 of ref.~\cite{Busch:2014h} these are perfectly accurate measurements.  But they will, of course, fail to reveal any  correlations between the two particles.

 Outside of the three examples under discussion Busch and his co-workers adopt a state-independent version of the principle, according to which a measurement is perfectly accurate if it perfectly reproduces the initial state distribution \emph{for every  initial state}.  The phrase in italics makes a crucial difference, as can be seen from the following modified version of the doctor scenario (originally suggested by Poulin~\cite{Poulin:2014})
\begin{quotation}
Alice takes 10 cities, in each of which the incidence of HIV is different.  She then takes a sample of 100 people from each of these cities and presents them to Bob for testing; without, however, telling Bob which patient comes from which city.  It turns out the proportion of positive test results for each city coincides with the actual proportion of HIV infected people in that city.  Alice concludes that, whatever it is that Bob is doing, it probably deserves to be considered a test.
\end{quotation}
Similarly with the state-independent version of the Busch \emph{et al} principle:  If a measurement reproduces the initial state distribution for \emph{every} choice of state then it is very plausible to argue that it is, in some sense, highly accurate.  Calculation confirms that impression.  In particular, it is easily seen that a measurement for which the state-independent Wasserstein 2-error is zero will successfully reveal the correlations in a Bell experiment.

However, in the examples under discussion Busch and his co-workers adopt a state-\emph{dependent} version of the principle.  Like the state-dependent version of the operator approach  this version of the principle can easily lead to unreasonable conclusions (\emph{c.f.} the broken-ammeter scenario in Section~\ref{sec:RMS}).  To show that their objection is not valid we will focus on example 3 in ref.~\cite{Busch:2014h}.  The extension to the other two examples will, we hope, be apparent.  We have already discussed this example at the end of Section~\ref{sec:Justification} (specializing to the case of a spin measurement).  As we noted there the \OP error is non-zero if the initial system state is not an eigenstate of $\hat{\sigma}_z$.  On the other hand, the fact that the initial system and apparatus states are the same, and the fact that the system and apparatus do not interact,  means that the distribution of measured values is identical to the initial state probability distribution of the measured observable.  Busch takes this to imply that the measurement is perfectly accurate.  The fact that the example is a quantum version of the first doctor scenario may make one suspicious of this conclusion.  To see that the suspicion is justified observe that if the measurement really were perfectly accurate  the result of a second, von Neumann measurement of the system observable should be perfectly correlated with the result of the first. It is easily seen that that is not the case.  Indeed, one actually finds
\ea{
\sqrt{\la \Psi | (\hat{\mu}_{2} -\hat{\mu}_{1})^2 |\Psi \ra} &= \Delta^{\hat{\rho}}_{\rm{e}} \sigma_z, 
}
where $|\Psi\ra$  is the Schr{\"{o}}dinger picture state of the system and two apparatuses after the second measurement is completed,  $\hat{\mu}_1$, $\hat{\mu}_2$ are the two pointer observables, and $\Delta^{\hat{\rho}}_{\rm{e}} \sigma_z$ is the state-dependent \OP error.

It is (to say the least) questionable whether the process just described counts as a  measurement at all.  Yet not only the state-dependent Wasserstein 2-error, but also the state-dependent \DM and \CO errors are zero.  That is not a weakness of the definitions:  In all three cases the fact that the error is zero is a well-defined operational statement which happens to be true---just as Bob's statement, in the broken-ammeter scenario, happens to be true.    It does, however, illustrate the limitations of state-dependent definitions.  We argued in Section~\ref{sec:RMS} that state-dependent definitions have their uses.  However, they need to be used with caution.  In particular, a state-dependent error is not a figure of merit:  its smallness does not, by itself, mean that a measurement is in any sense ``good''.

At this point we ought to stress that, although their arguments are, as it seems to us, invalid, the point that Busch and his co-workers are trying to establish---that there are measurements which are highly accurate as judged by any reasonable operational criterion but for which the \OP error is large---could  be right.  In Section~\ref{sec:Justification} we showed that smallness of the \DM or \CO quantities is both necessary and sufficient for a measurement to be accurate or non-disturbing in a well-defined, operational sense.  But in the case of the state-independent \OP quantities we only established sufficiency.  In the state-\emph{dependent} case Korzekwa \emph{et al}~\cite{Korzekwa:2014} have shown that there are processes which are completely non-disturbing as judged by any reasonable operational criterion, but for which the state-dependent \OP disturbance is non-zero (see the discussion at the end of  Section~\ref{sec:Justification}). However, it remains an open question whether the same is true of the state-dependent \OP error.  The more challenging, and to our mind more important question, of what can be said regarding the state-\emph{independent} \OP errors and disturbances, also remains open.

\section{Conclusion}
\label{sec:conclusion}
In the Introduction we argued that we need to develop a unified theory of measurement, in which  classical measurements are seen as a limiting case of  quantum measurements, rather in the way that  Newtonian kinematics is a limiting case of relativistic kinematics.  In particular, we need an overarching quantum mechanical concept of measurement accuracy which effectively reduces to the classical one in special cases, such as measurements with a meter rule.   We argued that, contrary to initial appearances, the Bell-KS theorem is not a major obstacle.

We made a start on the problem by showing that there are at least two ways to reformulate the classical definitions of error and disturbance in a way which does not involve a comparison with pre-existing values.  The reformulated definitions have natural quantum generalizations which we called the \DM and \CO definitions.  The \DM and \CO definitions are examples of quantum definitions which reduce to the classical concept in special cases.  They also bound the \OP quantities introduced in  refs.~\cite{Appleby:1998b,Ozawa:2003,Ozawa:2003a}.  They thereby give physical meaning to the \OP quantities.

We then turned to BLW's criticisms of the \OP definitions.  We argued that one should not expect there to be a single, canonical way of quantifying the concepts of measurement accuracy and disturbance.  The answer to the question ``what is the most appropriate quantitative definition?'' is always relative to the physical problem of interest.  We specified a class of problems for which the \DM definitions, and consequently the \OP definitions, are more appropriate than BLW's definitions based on the Wasserstein 2-deviation.  

Our analysis raises a number of questions which it might be interesting to investigate. Firstly, in the state-independent case, one would like to know (1)  whether smallness of the \OP quantities is  necessary and sufficient for the corresponding \DM and \CO quantities both to be small and (2) whether smallness of the \DM quantities is necessary and sufficient for the corresponding \CO quantities to be small.  Secondly, we have seen that there are physical problems for which the \DM quantities are better-suited than the ones based on the Wasserstein 2-deviation.  One would like to know if the same is true of the \CO quantities.   Thirdly it would be interesting to see if the \OP quantities capture any other operationally well-defined feature of the measurement, additional to those captured by the \DM and \CO quantities.  Finally, it would be interesting to see if one can prove error-disturbance and error-error relations expressed  in  terms of the \DM and \CO quantities.

\section*{Acknowledgements}
We are grateful to Paul Busch, Pekka Lahti, Masanao Ozawa and David Poulin for illuminating discussions.   This work was supported by the ARC via EQuS project number CE11001013.

\appendix

\section{Proof of Inequalities~\eqref{eq:MyErrDisIdeal}, \eqref{eq:MyErrErrIdeal}, \eqref{eq:MyErrDisReal}, \eqref{eq:MyErrErrReal}}
\label{sec:Technical}
In ref.~\cite{Appleby:1998b} we gave a proof of Inequalities~\eqref{eq:MyErrDisIdeal}, \eqref{eq:MyErrErrIdeal}, \eqref{eq:MyErrDisReal}, \eqref{eq:MyErrErrReal}  which skated over certain questions of domain and differentiability.   Our proof has since been criticized for its lack of rigour.  We here take the opportunity to fill in the missing details.  We incidentally strengthen the statement of \eqref{eq:MyErrDisReal}, \eqref{eq:MyErrErrReal} slightly.

The problem we face is that the operators $\hat{x}$, $\hat{p}$ are not defined on the whole Hilbert space.  Specifically, $|\psi\ra$ is in the domain of $\hat{x}$ if and only if the function $x^2 |\la x | \psi\ra|^2$ is integrable. The domain of $\hat{p}$ is defined similarly.  Expressed more intuitively:  $|\psi\ra$ is  in the domain of $\hat{x}$ if and only if $\la x^2 \ra < \infty$.  We take the view that  states  for which that is not true are never realized in a real, Earthbound laboratory experiment.  To put it more succinctly:  They are unphysical\footnote{Perhaps it will be speculated that such states may exist in nature.  We do not particularly  deny that possibility.  We only claim that they are irrelevant to the considerations of this paper.}.

One's first thought may be that one can define the domain of physical states to be the set of all $|\psi\ra$ such that $\la x |\psi\ra$ (respectively $\la p |\psi\ra$) is zero for all $x$ such that $|x| > B_{\rm{X}}$ (respectively, all $p$ such that $|p| > B_{\rm{P}}$), for suitably large positive constants $B_{\rm{X}}$, $B_{\rm{P}}$.  However, the set of such states is empty (because, if $\la x |\psi\ra$ is zero outside the interval $[-B_{\rm{X}}, B_{\rm{X}}]$, then its Fourier transform is analytic, by Schwartz's extension of the Paley-Wiener theorem [see, for example, Treves~\cite{Treves:1967}]).   Nevertheless, although the theory forces at least one of the wave-functions $\la x |\psi\ra$, $\la p |\psi\ra$ to have an infinite tail, nothing observable under ordinary laboratory conditions can depend on it.  It is not possible that a currently performable laboratory experiment can give rise to a state in which there is significant probability of the momentum being greater than $10^{(10^{10})}\rm{kg m s^{-1}}$, and even if it was possible one would not use non-relativistic quantum mechanics to describe it.  Nor is it possible to produce states for which there is significant probability of $|x|$ being greater than $10^{(10^{10})} \rm{m}$.  We need  to give quantitative expression to this point, that the infinite tails are physically irrelevant.  We accordingly take the view (inspired by the rigged Hilbert space formulation of quantum mechanics~\cite{Gelfand:1964a,Bohm:1978,Bohm:1989}) that the set of physical pure states $\mcl{P}_{0}$ consists of those states  $|\psi\ra$ for which the position space wave function $\la x |\hat{\rho} | y\ra$ is (a) $C^{\infty}$ and (b) rapidly decreasing at infinity in the sense that
\ea{
\sup_{x \in \fd{R}} \left( |1+|x|^2)^n \left| \frac{\partial^{m}}{\partial x^m} \la x |\psi \ra\right| \right) < \infty
}
for every pair of non-negative integers $n$, $m$ (in other words, $\la x |\psi\ra$ is a test function for the space of tempered distributions~\cite{Treves:1967}).  Note that this is equivalent to requiring that the momentum space wave function is $C^{\infty}$ and rapidly decreasing at infinity.   Note also that $\mcl{P}_{0}$ is in the domain of every monomial in $\hat{x}$ and $\hat{p}$.    

At first sight this definition may appear arbitrary.  The reader may allow that it is reasonable to impose some restriction on the behaviour at infinity, but wonder why we make this particular choice. Indeed, the requirement is much stronger than we need for our present purposes.   We make the definition nonetheless because actually it is not arbitrary, as can be seen by using  a von Neumann lattice\cite{Neumann:1932,Perelomov:1971,Bargmann:1971}.  For some appropriate scale-length $\lambda$,  let $|n,m\ra$ be the coherent state with wave function
\ea{
\la x | n,m\ra &= \left(\frac{1}{\pi \lambda^2}\right)^{\frac{1}{4}} e^{-\frac{1}{2\lambda^2} (x-n \lambda )^2 + \frac{2\pi m i}{\lambda}x }.
}
Then the set $\{|n,m\ra \colon n,m\in \fd{Z}\}$ with one point removed is a basis.  Choose some suitably enormous integer $N$ and let  $\hat{P}$ be the projector onto the finite dimensional subspace spanned by the set $\{|n,m\ra \colon -N \le n,m \le N\}$.  Then for any state $|\psi\ra$ that is relevant to  a real laboratory experiment the quantity $\| (1-\hat{P}) |\psi\ra \|)$ will be  negligible.  Consequently, predictions obtained using the state $|\psi\ra$ will be experimentally indistinguishable from ones obtained using the state
\ea{
 |\psi_r\ra=\frac{1}{\| \hat{P} |\psi\ra \|} \hat{P} |\psi\ra.
}
Without loss of predictive power we may therefore replace $|\psi\ra$ with $|\psi_r\ra$.  
The fact that $|\psi_r\ra$ is a finite linear combination of coherent states means  that it belongs to $\mcl{P}_0$.    Of course, $\mcl{P}_0$ also includes states like $|n,m\ra$ 
with $n$, $m$ both much larger than $N$, which are certainly not relevant to  ordinary, Earthbound laboratory experiments (being  localized outside the cosmic event horizon).  The point is only that every pure state which is experimentally relevant is empirically indistinguishable from  a state in $\mcl{P}_0$.

Finally we need to define $\mcl{P}$, the set of physical density matrices.  For each non-negative integer $m$ and real $\beta$ define norms
\ea{
N_{m,\beta} (|\psi\ra) = \sup_{x\in \fd{R}} \left( (1+|x|^2)^{\beta} \left| \frac{\partial^m}{\partial x^m}  \la x | \psi \ra \right| \right),
\nn
\tilde{N}_{m,\beta} (|\psi\ra) = \sup_{p\in \fd{R}} \left( (1+|p|^2)^{\beta} \left| \frac{\partial^m}{\partial p^m}  \la p | \psi \ra \right| \right).
}
Now let 
\ea{
\hat{\rho} &= \sum_n \xi_n |n\ra \la n| 
\label{eq:genDens}
} 
be an arbitrary density matrix with eigenvectors $|n\ra$.  We define   $\mcl{P}$ to consist of those $\hat{\rho}$ for which $|n\ra \in \mcl{P}_0$ for all $n$ and for which
\ea{
\sup_n \left( N_{m,\beta}(|n\ra)\right) & < \infty &  \sup_n \left( \tilde{N}_{m,\beta}(|n\ra)\right) &< \infty
}
for all $m$, $\beta$.  Note that it is enough to demand that one set of suprema is finite, since the finiteness of the other is then automatic.  Note also that in the case when the spectrum of $\hat{\rho}$ has degeneracies the finiteness of the suprema does not depend on the particular choice of eigenvectors.  Finally, let us remark that for the technical purposes of this appendix it would be enough to require that the suprema are finite for the particular case $m=0$, $\beta = 3/2$. 

This definition is justified by the fact that no experimentally relevant density matrix can be distinguished empirically from  a state in $\mcl{P}$.  Indeed, let $\hat{\rho}$ be  an experimentally relevant density matrix, and let $\hat{P}$ be the projector onto the first $N$ eigenstates.  Choose $N$ so that $\Tr((1-\hat{P}) \hat{\rho})$ is smaller than some suitably tiny number.  No practicable experiment can distinguish between $\hat{\rho}$ and $(1/(\Tr(\hat{P}\hat{\rho}))\hat{P}\hat{\rho}\hat{P}$.  By the argument we used  to justify the definition of $\mcl{P}_0$, the state $(1/(\Tr(\hat{P}\hat{\rho}))\hat{P}\hat{\rho}\hat{P}$ is in turn empirically indistinguishable from one of the form
\ea{
\hat{\rho}_0 &= \sum_n \xi_n |n\ra \la n | 
}
where the states $|n\ra$ are a finite orthonormal set  in $\mcl{P}_0$.  The fact that the set is finite means $\hat{\rho}_0\in \mcl{P}$.

The proof of the main theorem depends on three technical lemmas.  Define displacement operators
\ea{
\hat{D}_{xp} &= e^{i(p\hat{x} - x \hat{p})},
} 
and, for each $|\psi\ra \in \mcl{P}_0$, let
\ea{
|\psi_{xp}\ra &= \hat{D}\vpu{\dagger}_{xp} |\psi\ra.
}
It is easily seen that $|\psi_{xp}\ra  \in \mcl{P}_0$ for all $x$, $p$. 
\begin{Lemma}
\label{lm:Difflem}
For all $|\psi\ra \in\mcl{P}_0$ the function $|\psi_{xp}\ra$ is differentiable in the sense that\ea{
\left\| \frac{1}{\epsilon} \left(|\psi_{x+\epsilon,p}\ra -|\psi_{x,p}\ra\right)   + \frac{i}{\hbar}\left( \hat{p}-\frac{1}{2} p\right)|\psi_{x,p}\ra \right\| & \le \epsilon B (4+p^2) \tilde{N}_{0,\frac{3}{2}}(|\psi\ra),
\\
\left\| \frac{1}{\epsilon} \left(|\psi_{x,p+\epsilon}\ra -|\psi_{x,p}\ra\right)   - \frac{i}{\hbar}\left( \hat{x}-\frac{1}{2} x\right)|\psi_{x,p}\ra \right\| & \le \epsilon B(4+x^2) N_{0,\frac{3}{2}}(|\psi\ra)
}
for all $x$, $p$, all $\epsilon >0$, where $B$ is a fixed positive constant independent of  $|\psi\ra$.  
\end{Lemma}
\begin{proof}
It is easily seen that
\ea{
&\left| \frac{1}{\epsilon} \left(\la p' |\psi_{x+\epsilon,p}\ra -\la p'|\psi_{x,p}\ra\right)   + \frac{i}{\hbar}\left( p' - \frac{1}{2} p\right)\la p' |\psi_{x,p}\ra\right|^2 
\nn
&\hspace{0.5 in } \le \frac{ \epsilon^2 C}{16\hbar^4}  (2p'-p)^4 \left| \la p'-p |\psi \ra\right|^2
}
where
\ea{
C &= \sup_{u \in \fd{R}} \left( \frac{(\cos u-1)^2}{u^4} + \frac{(\sin u -u)^2}{u^4} \right).
}
Hence
\ea{
&\left\| \frac{1}{\epsilon} \left(|\psi_{x+\epsilon,p}\ra -|\psi_{x,p}\ra\right)   + \frac{i}{\hbar}\left( \hat{p}-\frac{1}{2} p\right)|\psi_{x,p}\ra \right\|^2
\nn
&\hspace{1 in} 
\le \frac{3 \pi \epsilon^2 C}{128 \hbar^2}(4+p^2)^2 \left(\tilde{N}_{0,\frac{3}{2}}(|\psi\ra) \right)^2.
}
The second inequality is proved in the same way.
\end{proof}
\begin{Lemma}
\label{cor:uniformCont}
For all $|\psi\ra \in \mcl{P}_0$, $|\psi_{xp}\ra$, $\hat{x} |\psi_{xp}\ra$, $\hat{p} |\psi_{xp}\ra$ are uniformly continuous  on every compact subset of $\fd{R}^2$.  Specifically, let $\mcl{C}$ be such a set.  Then
\ea{
\bigl\| |\psi_{x_1,p_1} \ra - |\psi_{x_2,p_2} \ra \bigr\|
& \le \epsilon_x B_1 \tilde{N}_{0,\frac{3}{2}}(|\psi\ra) + \epsilon_p B_1 N_{0,\frac{3}{2}} (|\psi\ra)
\label{eq:unContA}
\\
\bigl\| \hat{x} |\psi_{x_1,p_1} \ra - \hat{x} |\psi_{x_2,p_2} \ra \bigr\|
&\le
\epsilon_x B_2 \left(1+  \tilde{N}_{0,\frac{3}{2}}( |\psi\ra)+\tilde{N}_{0,\frac{3}{2}}(\hat{x} |\psi\ra) \right)
\nn
&\hspace{1 in} 
+ \epsilon_p B_2 \left(N_{0,\frac{3}{2}}(|\psi\ra)+N_{0,\frac{3}{2}}(\hat{x} |\psi\ra) \right)
\\
\bigl\| \hat{p} |\psi_{x_1,p_1} \ra - \hat{p} |\psi_{x_2,p_2} \ra \bigr\|
&\le
\epsilon_x B_3 \left( \tilde{N}_{0,\frac{3}{2}}( |\psi\ra)+\tilde{N}_{0,\frac{3}{2}}(\hat{p} |\psi\ra) \right)
\nn
&\hspace{1 in} 
+ \epsilon_p B_3 \left(1+N_{0,\frac{3}{2}}(|\psi\ra)+N_{0,\frac{3}{2}}(\hat{p} |\psi\ra) \right)
}
for all $(x_1,p_1)$, $(x_2,p_2)\in \mcl{C}$, where $\epsilon_x = |x_1-x_2|$, $\epsilon_p = |p_1-p_2|$, and where  the $B_j$ are  positive constants which depend on $\mcl{C}$ but not on $|\psi\ra$.
\end{Lemma}
\begin{proof}
The first inequality is a straightforward consequence of Lemma~\ref{lm:Difflem} and the inequalities
\ea{
 \left\| \left(\hat{x} - \frac{1}{2} x \right) |\psi_{x,p}\ra \right\|^2
 &\le
 \left(N_{0,\frac{3}{2}}(|\psi\ra)\right)^2 \frac{\pi}{32} (4+3 x^2),
 \\
\left\| \left(\hat{p} - \frac{1}{2}p \right) |\psi_{x,p}\ra \right\|^2
&\le  \left(\tilde{N}_{0,\frac{3}{2}}(|\psi\ra)\right)^2 \frac{\pi}{32} (4+3p^2).
}
To prove the second inequality let $|\phi\ra = \hat{x} |\psi\ra$.  Then
\ea{
\hat{x} |\psi_{xp}\ra &= |\phi_{xp} \ra + x |\psi_{xp}\ra,
}
implying
\ea{
&\bigl\| \hat{x} |\psi_{x_1,p_1} \ra - \hat{x} |\psi_{x_2,p_2} \ra \bigr\|
\nn
&\hspace{0.5 in} \le \bigl\| |\phi_{x_1,p_1} \ra -|\phi_{x_2,p_2}\ra \bigr\| + x_1 \bigl\| \psi_{x_1,p_1} \ra - |\psi_{x_2,p_2}\ra\bigr\| + |x_1-x_2|.
}
The proof now reduces to an application of the first inequality.  The last inequality is proved in the same way.
\end{proof}
Let the initial apparatus state be
\ea{
\ap &= \sum_{n=1}^{n_{\rm{a}}} \lambda_n |\phi_n \ra \la \phi_n |
\label{eq:AppsigExp}
}
for some set of positive numbers $\lambda_n$ and orthonormal  set $|\phi_n\ra$.  We argue on the same physical grounds adduced in the first few paragraphs of this appendix that the quantities $\Tr\bigl((\sy \otimes \ap)\eop{x}^2\bigr)$, $\Tr\bigl(\sy \otimes \ap) \dop{p}^2\bigl)$ are well-defined for all $\sy \in \mcl{P}$. 
 Finally, for given positive real numbers $l_{\rm{X}}$, $l_{\rm{P}}$, define $\mcl{C}_{l_{\rm{X},} l_{\rm{P}}}$ to be the phase-space box consisting of all $x,p$ such that  $ -\frac{l_{\rm{X}}}{2} \le x \le \frac{l_{\rm{X}}}{2}$, $ -\frac{l_{\rm{P}}}{2} \le x \le \frac{l_{\rm{P}}}{2}$.  
\begin{Lemma}
\label{lm:expctDiff}
Let $\hat{\rho}$ be any element of  $\mcl{P}$, and 
let $\hat{\rho}_{xp}= \hat{D}\vpu{\dagger}_{xp} \hat{\rho} \hat{D}^{\dagger}_{xp}$.  Let $\hat{A}$ be any self-adjoint operator such that $\Tr\bigl(\hat{\rho}_{xp} \otimes \ap)[\hat{p},\hat{A} ] \bigr)$, 
$\Tr\bigl(\hat{\rho}_{xp} \otimes \ap)[\hat{x},\hat{A}] \bigr)$ 
are defined, 
and $\Tr\bigl( (\hat{\rho}_{xp} \otimes \ap) \hat{A}^2 \bigr)$, $\Tr\bigl( (\hat{\rho}_{xp} \otimes \ap)\hat{x} \hat{A}^2 \hat{x} \bigr)$, $\Tr\bigl( (\hat{\rho}_{xp} \otimes \ap)\hat{p} \hat{A}^2 \hat{p} \bigr)$ are both defined and bounded on $ \mcl{C}_{l_{\rm{X}},l_{\rm{P}}}$.
Then $\Tr\bigl( (\hat{\rho}_{xp} \otimes \ap) \hat{A} \bigr)$ is  differentiable on $ \mcl{C}_{l_{\rm{X}},l_{\rm{P}}}$, and
\ea{
\frac{\partial}{\partial x}\left( \Tr\bigl(\hat{\rho}_{xp} \otimes \ap)\hat{A} \right) 
&= \frac{i}{\hbar} \Tr\bigl(\hat{\rho}_{xp} \otimes \ap)[\hat{p},\hat{A} ] \bigr),
\label{eq:traDiffA}
\\
\frac{\partial}{\partial p}\left( \Tr\bigl(\hat{\rho}_{xp} \otimes \ap)\hat{A} \right) 
&= -\frac{i}{\hbar} \Tr\bigl(\hat{\rho}_{xp} \otimes \ap)[\hat{x},\hat{A}] \bigr).
\label{eq:traDiffB}
}
Moreover, the derivatives are uniformly continuous on $ \mcl{C}_{l_{\rm{X}},l_{\rm{P}}}$.
\end{Lemma}
\begin{proof}
We have
\ea{
\hat{\rho} &= \sum_{n=1}^{n_{\rm{s}}} \xi_n |\psi_n \ra\la \psi_n |
}
for some set of positive numbers $\xi_n$ and orthonormal vectors $|\psi_n\ra \in \mcl{P}_0$.  
 Define
\ea{
f_{n,m}(x,p) &= \left\| \hat{A} |\psi_{n,x,p} \otimes \phi_m\ra\right\|,
\\
g_{n,m}(x,p) &= \left| \hat{A} \hat{x} |\psi_{n,x,p} \otimes \phi_m \ra \right\|,
\\
h_{n,m}(x,p) &= \left\| \hat{A} \hat{p} |\psi_{n,x,p} \otimes \phi_m \ra \right\|,
}
\ea{
C_1 &= \sup_{x,p \in \mcl{C}_{l_{\rm{X}},l_{\rm{P}}}} \left(\Tr\bigl( (\hat{\rho}_{xp} \otimes \ap) \hat{A}^2 \bigr) \right),
\\
C_2 &= \sup_{x,p \in \mcl{C}_{l_{\rm{X}},l_{\rm{P}}}} \left( \Tr\bigl( (\hat{\rho}_{xp} \otimes \ap)\hat{x} \hat{A}^2 \hat{x} \bigr)  \right),
\\
C_3 &= \sup_{x,p \in \mcl{C}_{l_{\rm{X}},l_{\rm{P}}}} \left( \Tr\bigl( (\hat{\rho}_{xp} \otimes \ap)\hat{p} \hat{A}^2 \hat{p} \bigr) \right),
}
\ea{
C_4 &= \sup_n \bigl(N_{0,\frac{3}{2}}\bigl(|\psi_n\ra\bigr) \bigr), & \tilde{C}_4 &= \sup_n \bigl(\tilde{N}_{0,\frac{3}{2}}\bigl(|\psi_n\ra\bigr) \bigr),
\\
C_5 &= \sup_n \bigl(N_{0,\frac{3}{2}}\bigl(\hat{x}|\psi_n\ra\bigr) \bigr), & \tilde{C}_5 &= \sup_n \bigl(\tilde{N}_{0,\frac{3}{2}}\bigl(\hat{x} |\psi_n\ra\bigr) \bigr),
\\
C_6 &= \sup_n \bigl(N_{0,\frac{3}{2}}\bigl(\hat{p} |\psi_n\ra\bigr) \bigr), & \tilde{C}_6 &= \sup_n \bigl(\tilde{N}_{0,\frac{3}{2}}\bigl(\hat{p} |\psi_n\ra\bigr) \bigr).
\label{eq:C6Def}
}
It follows from Lemmas~\ref{lm:Difflem} and~\ref{cor:uniformCont} that
\ea{
\phantom{\le}&\left|\frac{1}{\epsilon} \bigl( \la \psi_{n,x+\epsilon,p} \otimes \phi_m | \hat{A} | \psi_{n,x+\epsilon,p} \otimes \phi_m \ra -
\la \psi_{n,x,p} \otimes \phi_m | \hat{A} | \psi_{n,x,p} \otimes \phi_m \ra \bigr) \right.
\nn
&\hspace{1.5 in} \left.  - \frac{i}{\hbar} \la \psi_{n,x,p} \otimes \phi_m | [\hat{p},\hat{A}] | \psi_{n,x,p}\otimes \phi_m \ra 
\right|
\nn
&  \le 2\epsilon\tilde{C}_4B(4+p^2)f_{n,m}(x,p) +\frac{\epsilon}{\hbar} B_1\tilde{C}_4 \left( h_{n,m}(x,p) + \frac{|p|}{2}  f_{n,m}(x,p)\right).
}
It follows from the Cauchy-Schwartz inequality that
\ea{
\sum_{n,m} \xi_n \lambda_m f_{n,m}(x,p) &\le \sqrt{C_1},  & \sum_{n,m} \xi_n \lambda_m h_{n,m}(x,p) & \le \sqrt{C_3}.
}
Consequently
\ea{ 
&\left|\frac{1}{\epsilon} \Bigl( \Tr\bigl((\hat{\rho}_{x+\epsilon,p} \otimes \ap) \hat{A} \bigr) - 
\Tr\bigl((\hat{\rho}_{x,p} \otimes \ap) \hat{A} \bigr)\Bigr) - \frac{i}{\hbar} \Tr\bigl( (\hat{\rho}_{x,p}\otimes \ap) [\hat{p},\hat{A}]\bigr)\right|
\nn
& \le 
2\epsilon\sqrt{C_1} \tilde{C}_4B(4+p^2) +\frac{\epsilon}{\hbar} B_1\tilde{C}_4 \left(\sqrt{C_3 }+ \frac{|p|}{2}  \sqrt{C_1}\right)
\nn 
& \to 0
}
as $\epsilon \to 0$, which establishes Eq.~\eqref{eq:traDiffA}.  Eq.~\eqref{eq:traDiffB} is proved similarly.

Lemmas~\ref{lm:Difflem} and~\ref{cor:uniformCont} also  imply
\ea{
&\Bigl| \la \psi_{n,x+\epsilon_1,p+\epsilon_2} \otimes \phi | [\hat{p},\hat{A} ] | \psi_{n,x+\epsilon_1,p+\epsilon_2}\ra - \la \psi_{n,x,p} \otimes \phi | [\hat{p},\hat{A} ] | \psi_{n,x,p}\ra \Bigr| 
\nn
&\hspace{0.5 in} \le 2h_{n,m}(x,p) B_1 \left( \epsilon_1 \tilde{C}_4 + \epsilon_2 C_4\right)
\nn
&\hspace{1 in} 
+2f_{n,m}(x,p) B_3\left( \epsilon_1(\tilde{C}_4+\tilde{C}_6)+ \epsilon_2(1+ C_4+C_6)\right).
}
Uniform continuity of the $x$ derivative now follows by another application of the Cauchy-Schwartz inequality.  Uniform continuity of the $p$ derivative is proved similarly.
\end{proof}
We are now ready to prove our main result.
\begin{Theorem}
\label{thm:mainresult}
Let $\mcl{R}$ be a subset of $\mcl{P}$ containing at least one state $\rho$ such that 
\ea{
&\hat{\rho}_{xp}, \;\; \frac{1}{\Tr(\hat{\rho}_{xp} \hat{x}^2)} \hat{x} \hat{\rho}_{xp} \hat{x}, \;\;
\frac{1}{\Tr(\hat{\rho}_{xp} \hat{p}^2)} \hat{p} \hat{\rho}_{xp} \hat{p} \;\;  \in \mcl{R}
\label{eq:mainResCond}
}
for all $(x,p) \in \mcl{C}_{l_{\rm{X}},l_{\rm{P}}}$.  Then
\ea{
\evc{x} \dvc{p} + \frac{\hbar}{l_{\rm{X}}} \evc{x} + \frac{\hbar}{l_{\rm{P}}} \dvc{p} \ge \frac{\hbar}{2}
}
for every measurement of position, and
\ea{
\evc{x} \evc{p} + \frac{\hbar}{l_{\rm{X}}} \evc{x} + \frac{\hbar}{l_{\rm{P}}} \evc{p} \ge \frac{\hbar}{2}
}
for every joint measurement of position and momentum.
\end{Theorem}
\begin{proof}
To prove the first relation observe that it is automatic if either of the quantities $\evc{x}$, $\dvc{p}$ is infinite.  Suppose, on the other hand, they are both finite.  It is easily seen that 
\ea{
[\eop{x},\dop{p}] &= -i\hbar - [\hat{x}_{\rm{i}},\dop{p}] +[\hat{p}_{\rm{i}},\eop{x}].
}
Let $\hat{\rho}$ be any state in $\mcl{R}$ satisfying  condition~\eqref{eq:mainResCond}.  It is easily seen that $\Tr\bigr(\hat{\rho}_{xp}\hat{x}^2\bigr)$, $\Tr\bigr(\hat{\rho}_{xp} \hat{p}^2\bigr)$ are bounded.  So we can apply Lemma~\ref{lm:expctDiff} with $\hat{A}=\eop{x}, \dop{p}$ to deduce
\ea{
\Tr\bigl((\hat{\rho}_{xp} \otimes \ap )[\eop{x},\dop{p}]\bigr) &= -i\hbar (1 + \boldsymbol{\nabla}\cdot \mbf{v})
}
where $\boldsymbol{\nabla} = \left(\begin{smallmatrix} \frac{\partial}{\partial x} \\ \frac{\partial}{\partial p} \end{smallmatrix}\right)$ and
\ea{
\mbf{v} &= \bmt \Tr\bigl((\hat{\rho}_{xp} \otimes \ap) \eop{x}\bigr) \\ \Tr\bigl((\hat{\rho}_{xp} \otimes \ap) \dop{p}\bigr)\emt .
}
Since $\boldsymbol{\nabla}\cdot \mbf{v}$ is continuous it is integrable.  Hence
\ea{
\evc{x} \dvc{p} & \ge \frac{\hbar}{2} \sup_{x,p\in \mcl{C}_{\lX,\lP}} \bigl(1+\boldsymbol{\nabla}\cdot \mbf{v}\bigr)
\nn
& \ge \frac{\hbar}{2} \left( 1-\frac{1}{\lX\lP} \left| \int_{\mcl{C}_{\lX,\lP}} \boldsymbol{\nabla}\cdot \mbf{v}\; dx dp \right| \right)
\nn
& =  \frac{\hbar}{2} \left( 1-\frac{1}{\lX\lP} \left| \int_{\mcl{B}_{\lX,\lP}}  \mbf{n} \cdot \mbf{v} \;  ds \right| \right)
\nn
&\ge \hbar\left( \frac{1}{2} - \frac{1}{\lX} \evc{x} - \frac{1}{\lP} \dvc{p}\right)
}
where $\mcl{B}_{\lX,\lP}$ is the boundary of $\mcl{C}_{\lX,\lP}$ and $\mbf{n}$ is the outward-pointing normal.
The second inequality is proved in the same way, starting from the commutation relation
\ea{
[\eop{x},\eop{p}] &= -i\hbar - [\hat{x}_{\rm{i}},\eop{p}] + [\hat{p}_{\rm{i}},\eop{x}].
}
\end{proof}
Inequalities \eqref{eq:MyErrDisIdeal}, \eqref{eq:MyErrErrIdeal} are proved by specializing to the case $\mcl{R} = \mcl{P}$ and taking the limit as $\lX$, $\lP\to \infty$.

\end{document}